\documentclass{svjour3}
\usepackage{cite}
\usepackage[utf8]{inputenc}
\usepackage{amssymb}
\usepackage{amsmath}
\usepackage{enumerate}
\usepackage{dsfont}
\usepackage{hyperref}
\usepackage{graphicx}
\usepackage{color}
\usepackage{tikz}
\usepackage{comment}
\usepackage{placeins}
\usepackage[]{algorithm2e}
\usetikzlibrary{calc}
\usetikzlibrary{shapes,snakes}

\newif\iflong
\longtrue

\newtheorem{Rule}{Rule}{\bfseries}{\rmfamily}
\newtheorem{Observation}{Observation}{\bfseries}{\rmfamily}
\newtheorem{Assumption}{Assumption}{\bfseries}{\rmfamily}

\newcommand{\Oh}{\mathcal{O}}
\newcommand{\badstuffhappens}{\ensuremath{\text{\rm NP} \subseteq \text{\rm coNP} / \text{\rm poly}}}
\DeclareMathOperator{\poly}{poly}

\begin{document}

\title{On the Relation of Strong Triadic Closure and Cluster~Deletion\footnote{A preliminary version of this work appeared in \emph{Proceedings  of the 44th International Workshop on Graph-Theoretic Concepts in Computer Science (WG~'18)}  held in Cottbus, Germany in June 2018~\cite{GK18}. The long version contains full proofs of all statements, an improved running time for the kernelization, and a slightly extended dichotomy for $H$-free graphs.}}
\author{Niels Grüttemeier and Christian Komusiewicz}

\institute{Fachbereich für Mathematik und Informatik, Philipps-Universität Marburg, Germany\\
\email{\{niegru,komusiewicz\}@informatik.uni-marburg.de}}
\maketitle
\begin{abstract}
  We study the parameterized and classical complexity of two problems that are concerned with induced paths on three vertices, called~$P_3$s, in
  undirected graphs~$G=(V,E)$.  In \textsc{Strong Triadic Closure} we aim to label the
  edges in~$E$ as strong and weak such that at most~$k$ edges are weak and~$G$ contains no
  induced~$P_3$ with two strong edges. In \textsc{Cluster Deletion} we aim to destroy all
  induced~$P_3$s by a minimum number of edge deletions. We first show that \textsc{Strong
    Triadic Closure} admits a~$4k$-vertex kernel. Then, we study parameterization
  by~$\ell:=|E|-k$ and show that both problems are fixed-parameter tractable and 
  unlikely to admit a polynomial kernel with respect to~$\ell$. Finally, we give a
  dichotomy of the classical complexity of both problems on~$H$-free graphs for all~$H$ of
  order at most four.
\end{abstract}
\section{Introduction}
We study two related graph problems arising in social network
analysis and data clustering. Assume we are given a social network where
vertices represent agents and edges represent relationships between these agents,
and want to predict which of the relationships are important. In online social networks for
example, one could aim to distinguish between close friends  and spurious
relationships.  Sintos~and~Tsaparas~\cite{sintosL} proposed to use the notion of strong triadic
closure for this problem. This notion goes back to Granovetter's sociological work \cite{Granovetter73}. Informally, it is the assumption that if one agent has strong relationships
with two other agents, then these two other agents 
should have at least a weak relationship. The aim in the computational problem 
is then to label a maximum number of edges of the given
social network as strong while fulfilling this requirement. 
Formally, we are looking for an \emph{STC-labeling} defined as follows.

\begin{definition}  A \emph{labeling}
  $L=(S_L,W_L)$ of an undirected graph~$G=(V,E)$ is a partition of the edge set~$E$. The edges in~$S_L$ are
  called \emph{strong} and the edges in~$W_L$ are called \emph{weak}. A labeling
  $L=(S_L,W_L)$ is an \emph{STC-labeling} if there exists no pair of strong edges
  $\lbrace u, v \rbrace \in S_L$ and $\lbrace v, w \rbrace \in S_L$ such that
  $\lbrace u, w \rbrace \not \in E$.
\end{definition}
For any weak (strong) edge $\lbrace u, v \rbrace$ we refer to $u$ as a weak (strong)
neighbor of~$v$.  The computational problem described informally above is now the
following.
\begin{center}
	\begin{minipage}[c]{.9\linewidth}
          \textsc{Strong Triadic Closure (STC)}\\
          \textbf{Input}: An undirected graph~$G=(V, E)$ and an
          integer~$k \in \mathbb{N}$.\\
          \textbf{Question}: Is there an STC-labeling $L=(S_L,W_L)$ with $|W_L| \leq k$?
	\end{minipage}
\end{center}
We call an STC-labeling $L$ \emph{optimal} for a graph $G$, if the number
$|W_L|$ of weak edges is minimal. The
STC-labeling property can also be stated in terms of induced subgraphs: For every
induced~$P_3$, the path on three vertices, of~$G$ at most one edge is labeled strong.
Therefore, as observed previously~\cite{KNP17L}, \textsc{STC} is closely related to the problem of
destroying induced~$P_3$s by edge deletions. Since the graphs without an induced~$P_3$ are
exactly the disjoint union of cliques, this problem is usually formulated as follows.
\begin{center}
	\begin{minipage}[c]{.9\linewidth}
          \textsc{Cluster Deletion (CD)}\\
          \textbf{Input}: An undirected graph~$G=(V, E)$ and an
          integer~$k \in \mathbb{N}$.\\
          \textbf{Question}: Can we transform~$G$ into a \emph{cluster graph}, that is, a
          disjoint union of cliques, by at most~$k$ edge deletions?
	\end{minipage}
\end{center}
More precisely, any set~$D$ of at most~$k$ edge deletions that transform~$G$ into a
cluster graph, directly implies an STC-labeling~$(E\setminus D,D)$ with at most~$k$ weak
edges. There are, however, graphs~$G$ where the minimum number of weak edges in an
STC-labeling is strictly smaller than the number of edge deletions that are needed in
order to transform~$G$ into a cluster graph~\cite{KNP17L}. Due to the close relation
between the two problems, there are graph classes where any minimum-cardinality
solution for \textsc{Cluster Deletion} directly implies an optimal
STC-labeling~\cite{KNP17L}.

In this work, we study the parameterized complexity of \textsc{STC} and \textsc{Cluster
  Deletion} and the classical computational complexity of both problems in graph classes
that can be described by one forbidden induced subgraph of order at most four.

\paragraph{Known Results.} \textsc{STC} is NP-hard~\cite{sintosL}, even when restricted to
graphs with maximum degree four~\cite{KNP17L} or to~split graphs~\cite{konstantinidis_et_alL}.
In contrast, \textsc{STC} is solvable in polynomial time when the input graph is
bipartite~\cite{sintosL}, subcubic~\cite{KNP17L}, a proper interval
graph~\cite{konstantinidis_et_alL}, or a cograph, that is, a graph with no
induced~$P_4$~\cite{KNP17L}. \textsc{STC} can be solved in~$\Oh(1.28^k + nm)$ time and admits a polynomial kernel when parameterized by~$k$. These two results follow from a parameter-preserving reduction to \textsc{Vertex Cover}, which asks if it is possible to delete at most $k$ vertices of a given graph such that the remaining graph does not contain any edge. This parameter-preserving reduction \cite{sintosL} computes the so-called Gallai graph \cite{Sun91, Bang96} of the input graph and directly gives the above-mentioned running time bound. The existence of a kernel for parameter~$k$ is implied by two facts: First, \textsc{Vertex Cover} admits a polynomial kernel for the number $k$ of vertices to delete~\cite{Cyg+15,DF13}. Second, \textsc{Vertex Cover} is in NP and \textsc{STC} is NP-hard. Hence, the \textsc{Vertex Cover} instance of size~$\poly(k)$ which we obtain by first reducing from STC to \textsc{Vertex Cover} and then applying the kernelization can be transformed into an equivalent \textsc{STC} instance by a polynomial-time reduction. The STC instance then has size~$\poly(k)$.

\textsc{Cluster Deletion} is NP-hard~\cite{SST04}, even when restricted to graphs with maximum
degree four~\cite{KU12} or to $(2K_2,3K_1)$-free graphs~\cite{GHN13}, and solvable in polynomial time on cographs~\cite{GHN13} and 
in time~$O(1.42^k+m)$ on general graphs~\cite{BD11}.

\paragraph{Our Results.}We provide the first linear-vertex kernel for 
\textsc{STC}
parameterized by~$k$. 
More precisely,
we show that in $\Oh(nm)$ time we can reduce an arbitrary instance of
\textsc{STC} to an equivalent instance with at most~$4k$ vertices. 
\begin{table}[t]
  \centering
\caption{The parameterized complexity of \textsc{STC} and \textsc{Cluster Deletion} for parameters~$k$ and~$\ell:=|E|-k$.}

\begin{tabular}{l c c}
\hline

Parameter & \textsc{STC} & \textsc{Cluster Deletion}\\
\hline
$k$& ~~~$\Oh(1.28^k+nm)$-time algo~\cite{sintosL}~~~ &
~~~$\mathcal{O}(1.42^k+m)$-time algo~\cite{BD11}~~~\\

& $4k$-vertex kernel (Thm.~\ref{thm:4k})
& $4k$-vertex kernel~\cite{Guo09}
\\

\hline
$\ell$ & $\ell^{\mathcal{O}(\ell)}\cdot n$-time algo (Thm.~\ref{thm:l-stc})
& $\Oh(9^{\ell}\cdot \ell n )$-time algo (Thm.~\ref{thm:l-cd})\\
& No poly kernel (Thm.~\ref{thm:no-kernel-stc}) & No poly kernel~(Cor.~\ref{cor:no-kernel-cd})\\
\hline



\end{tabular}
\label{tab:pc-results}
\end{table}

\begin{table}[t]
  \caption{Complexity Dichotomy and correspondence of \textsc{STC} and \textsc{Cluster Deletion} on $H$-free graphs.}

  \centering
\setlength{\tabcolsep}{4mm}
\begin{tabular}{p{5cm} c c c}
\hline
 & \textsc{STC} & \textsc{CD}& correspondent\\
\hline
$H \in \{3K_1, K_4, 4K_1,C_4,2K_2,\text{claw},$ & NP-h & NP-h & NO \\
\qquad $ \;\; \text{co-claw},\text{co-diamond},\text{co-paw}\}$ & & \\
$H= \text{diamond}$&NP-h&NP-h& YES\\
$H\in \{K_3,P_3,K_2K_1,\text{paw},P_4\}$ & P & P & YES\\
\hline
\end{tabular}

\label{tab:h-free-results}
\end{table}

Second, we initiate the study of the parameterized complexity of \textsc{STC} and
\textsc{Cluster Deletion} with respect to the parameter~$\ell:=|E|-k$. Hence,
in~\textsc{STC} we are searching for an STC-labeling with at least~$\ell$ strong edges and
in \textsc{Cluster Deletion} we are searching for a cluster graph~$G'$ that is a subgraph
of~$G$ and that has at least~$\ell$ edges; we call these edges the \emph{cluster edges} of~$G'$. While we present fixed-parameter algorithms for both problems and the
parameter~$\ell$, we also show that, somewhat surprisingly, both problems do not admit a
polynomial kernel with respect to~$\ell$, unless~\badstuffhappens. Our
result is obtained by polynomial parameter transformations from \textsc{Clique}
parameterized by the size of a vertex cover of the input graph to~\textsc{Multicolored
  Clique} parameterized by the sum of the sizes of all except one color class to
\textsc{STC} and \textsc{Cluster Deletion} parameterized by~$\ell$. The
\textsc{Multicolored Clique} variant may be of independent interest as a suitable base
problem for polynomial parameter transformations. Table~\ref{tab:pc-results} gives
an overview of the parameterized complexity.

Finally, we extend the line of research studying the complexity of \textsc{Cluster
  Deletion}~\cite{GHN13} and \textsc{STC}~\cite{KNP17L} on $H$-free graphs where~$H$ is a
graph of order at most four. We present a complexity dichotomy between polynomial-time solvable
and NP-hard cases for all possibilities for~$H$. Moreover, we show for all such graphs~$H$
whether~\textsc{STC} and \textsc{Cluster Deletion} \emph{correspond} on~$H$-free graphs,
that is, whether every STC-labeling with at most $k$~weak edges implies a \textsc{Cluster
  Deletion} solution with at most~$k$ edge deletions. These results are shown in
Table~\ref{tab:h-free-results}.

\paragraph{Related Work.}
Independent from our work, Golovach et al.~\cite{Pinar} showed that \textsc{STC} parameterized by $\ell$ has no polynomial kernel unless \badstuffhappens, even when the input graph is a split graph. Moreover, they discuss the \textsc{Strong $F$-Closure} problem---a generalization of \textsc{STC}. For an arbitrary graph $F$, the \textsc{Strong $F$-Closure} problem asks for a labeling $L=(S_L,W_L)$ of a graph $G$ such that there is no induced subgraph~$F$ in $G$ that contains only strong edges under $L$ and where the number of strong edges is maximum under this property. Among other results, it is shown that~\textsc{Strong $F$-Closure} admits a polynomial kernel for the parameter~$k$~\cite{Pinar}.

Several further problems that are closely related to STC have been studied recently. Sintos and Tsaparas~\cite{sintosL} introduced \textsc{Multi-STC}, a generalization of
STC where the labeling is allowed to have~$c$ different strong colors and for each
induced~$P_3$ in~$G$, the edges of~$G$ must be labeled by different strong colors or one
of the edges must be labeled weak. This variant is harder than STC in the sense that for
all~$c\ge 3$, \textsc{Multi-STC} is NP-hard even if~$k=0$~\cite{BGKS19}. A further
approach for predicting strong relationships based on strong triadic closure
asks for an STC-labeling of~$G=(V,E)$ such that each community~$X_i\subseteq V$ of a set
of given communities~$\{X_1,\ldots,X_t\}$ is internally connected via strong links~\cite{RTG17}. This problem has  been shown to be a special case of a facility location problem~\cite{BTZ19}.

\section{Preliminaries}

\paragraph{Graph Theory.}
We consider undirected simple graphs $G=(V,E)$ where $n:= |V|$ and $m := |E|$. For any vertex
$v \in V$, the open neighborhood of $v$ is denoted by $N_G(v)$, the closed neighborhood is
denoted by $N_G[v]$. The set of vertices in $G$ which have a distance of exactly $2$ to $v$ is
denoted by $N^2_G (v)$. For any two vertex sets $V_1,V_2 \subseteq V$, we let
$E_G(V_1,V_2) := \lbrace \lbrace v_1, v_2 \rbrace \in E \mid v_1 \in V_1, v_2 \in V_2 \rbrace$
denote the set of edges between $V_1$ and $V_2$. For any vertex set $V' \subseteq V$, we
let $E_G(V') := E_G(V',V')$ be the set of edges between the vertices of $V'$. We may omit the
subscript~$G$ if the graph is clear from the context.
For any $V' \subseteq V$, $G[V'] := ( V', E(V') )$ denotes the \emph{subgraph of~$G$ induced
  by}~$V'$.  

A \emph{clique} in a graph $G$ is a set $K \subseteq V$ of vertices such
that~$G[K]$ is complete. A \emph{cut} $C= (V_1, V_2)$ is a partition of the vertex set into
two parts. The \emph{cut-set} $E_C := E_G(V_1,V_2)$ is the set of edges between $V_1$ and
$V_2$. A \textit{matching} $M \subseteq E$ is a set of pairwise disjoint edges. A matching in a graph~$G$ is
\emph{maximal} if adding any edge to $M$ does not give a matching and \emph{maximum}
if~$G$ has no larger matching.
A graph~$G$ is~\emph{$H$-free} if it does not contain an induced subgraph that is
isomorphic to the graph~$H$. The small graphs used in this work are shown in Fig.~\ref{fig:small-graphs}. For further background on graph classes and their definition via forbidden induced subgraphs refer to~\url{http://graphclasses.org} or to the monograph by Brandstädt et al.~\cite{BLS99}.

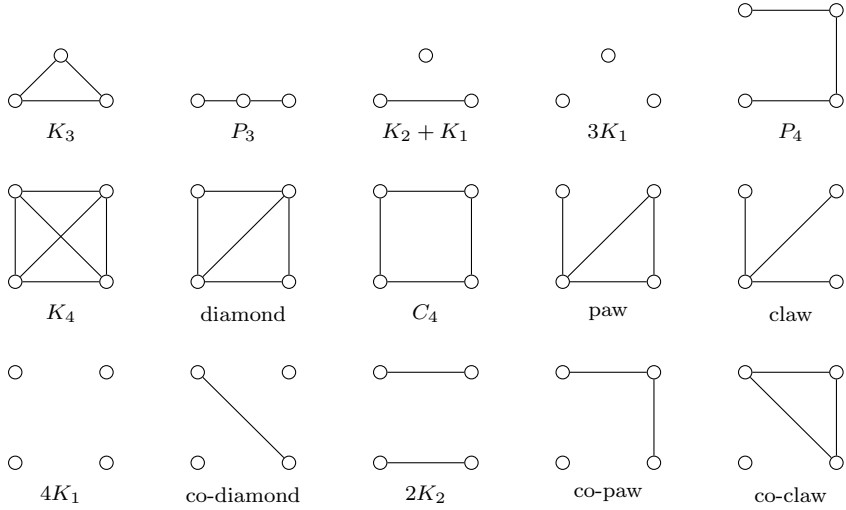
\begin{figure}[t]
  \begin{center}
    \begin{tikzpicture}[scale=0.6]
      \tikzstyle{knoten}=[circle,fill=white,draw=black,minimum size=5pt,inner sep=0pt]

      \begin{scope}

        \node[knoten] (1) at (0,0) {};
        \node[knoten] (2) at (2,0) {};
        \node[knoten] (3) at (1,1) {};
        \node at (1,-0.7) {$K_3$};
        
        \draw[-] (1) edge (2);
        \draw[-] (1) edge (3);
        \draw[-] (2) edge (3);
    \end{scope}
    \begin{scope}[xshift=4cm]

        \node[knoten] (1) at (0,0) {};
        \node[knoten] (2) at (1,0) {};
        \node[knoten] (3) at (2,0) {};
        \node at (1,-0.7) {$P_3$};
        
        \draw[-] (1) edge (2);
        \draw[-] (2) edge (3);
    \end{scope}

        \begin{scope}[xshift=8cm]

        \node[knoten] (1) at (0,0) {};
        \node[knoten] (2) at (1,1) {};
        \node[knoten] (3) at (2,0) {};
        \node at (1,-0.7) {$K_2+K_1$};
        
        \draw[-] (1) edge (3);
    \end{scope}

    \begin{scope}[xshift=12cm]

        \node[knoten] (1) at (0,0) {};
        \node[knoten] (2) at (1,1) {};
        \node[knoten] (3) at (2,0) {};
        \node at (1,-0.7) {$3K_1$};
        
    \end{scope}
    \begin{scope}[xshift=16cm]
        \node[knoten] (1) at (0,0) {};
        \node[knoten] (2) at (2,0) {};
        \node[knoten] (3) at (2,2) {};
        \node[knoten] (4) at (0,2) {};
        \node at (1,-0.7) {$P_4$};
        
        \draw[-] (1) edge (2);
        \draw[-] (2) edge (3);
        \draw[-] (3) edge (4);
    \end{scope}

    \begin{scope}[yshift=-4cm]
      \begin{scope}[xshift=4cm]
        \node[knoten] (1) at (0,0) {};
        \node[knoten] (2) at (2,0) {};
        \node[knoten] (3) at (2,2) {};
        \node[knoten] (4) at (0,2) {};
        \node at (1,-0.7) {diamond};
        
        \draw[-] (1) edge (2);
        \draw[-] (1) edge (3);
        \draw[-] (1) edge (4);
        \draw[-] (2) edge (3);
        \draw[-] (3) edge (4);
        
    \end{scope}

      \begin{scope}[xshift=8cm]

        \node[knoten] (1) at (0,0) {};
        \node[knoten] (2) at (2,0) {};
        \node[knoten] (3) at (2,2) {};
        \node[knoten] (4) at (0,2) {};
        \node at (1,-0.7) {$C_4$};
        
        \draw[-] (1) edge (2);
        \draw[-] (2) edge (3);
        \draw[-] (3) edge (4);
        \draw[-] (1) edge (4);
    \end{scope}

    \begin{scope}[xshift=12cm]
        \node[knoten] (1) at (0,0) {};
        \node[knoten] (2) at (2,0) {};
        \node[knoten] (3) at (2,2) {};
        \node[knoten] (4) at (0,2) {};
        \node at (1,-0.7) {paw};
        
        \draw[-] (1) edge (2);
        \draw[-] (1) edge (3);
        \draw[-] (2) edge (3);
        \draw[-] (1) edge (4);
        
    \end{scope}
    \begin{scope}[xshift=16cm]
        \node[knoten] (1) at (0,0) {};
        \node[knoten] (2) at (2,0) {};
        \node[knoten] (3) at (2,2) {};
        \node[knoten] (4) at (0,2) {};
        \node at (1,-0.7) {claw};
        
        \draw[-] (1) edge (2);
        \draw[-] (1) edge (3);
        \draw[-] (1) edge (4);
        
    \end{scope}

    \begin{scope}
        \node[knoten] (1) at (0,0) {};
        \node[knoten] (2) at (2,0) {};
        \node[knoten] (3) at (2,2) {};
        \node[knoten] (4) at (0,2) {};
        \node at (1,-0.7) {$K_4$};
        
        \draw[-] (1) edge (2);
        \draw[-] (1) edge (3);
        \draw[-] (1) edge (4);
        \draw[-] (3) edge (2);
        \draw[-] (2) edge (4);
        \draw[-] (3) edge (4);
        
    \end{scope}

  \end{scope}

  \begin{scope}[yshift=-8cm]

    \begin{scope}[xshift=4cm]
        \node[knoten] (1) at (0,0) {};
        \node[knoten] (2) at (2,0) {};
        \node[knoten] (3) at (2,2) {};
        \node[knoten] (4) at (0,2) {};
        \node at (1,-0.7) {co-diamond};
        
        \draw[-] (2) edge (4);
        
    \end{scope}

    \begin{scope}[xshift=8cm]
      \node[knoten] (1) at (0,0) {};
      \node[knoten] (2) at (2,0) {};
      \node[knoten] (3) at (2,2) {};
      \node[knoten] (4) at (0,2) {};
      \node at (1,-0.7) {$2K_2$};
      
      \draw[-] (1) edge (2);
      \draw[-] (3) edge (4);
    \end{scope}

    \begin{scope}[xshift=12cm]

        \node[knoten] (1) at (0,0) {};
        \node[knoten] (2) at (2,0) {};
        \node[knoten] (3) at (2,2) {};
        \node[knoten] (4) at (0,2) {};
        \node at (1,-0.7) {co-paw};
        
        \draw[-] (2) edge (3);
        \draw[-] (3) edge (4);
    \end{scope}
    \begin{scope}[xshift=16cm]
        \node[knoten] (1) at (0,0) {};
        \node[knoten] (2) at (2,0) {};
        \node[knoten] (3) at (2,2) {};
        \node[knoten] (4) at (0,2) {};
        \node at (1,-0.7) {co-claw};
        
        \draw[-] (2) edge (3);
        \draw[-] (2) edge (4);
        \draw[-] (3) edge (4);
    \end{scope}

    \begin{scope}
      \node[knoten] (1) at (0,0) {};
      \node[knoten] (2) at (2,0) {};
      \node[knoten] (3) at (2,2) {};
      \node[knoten] (4) at (0,2) {};
      \node at (1,-0.7) {$4K_1$};
      \end{scope}
   
  \end{scope}        
\end{tikzpicture}
\end{center}
\caption{The small graphs~$H$ considered in this work.}
\label{fig:small-graphs}
\end{figure}

\paragraph{Parameterized Algorithmics.}
In parameterized algorithmics~\cite{Cyg+15,DF13}, one analyzes the complexity of problems depending on the input size $n$ and a problem parameter~$k$. For a given problem, the aim is to obtain \emph{fixed-parameter tractable} (FPT) algorithms, which are algorithms with running time $f(k) \cdot \text{poly}(n)$ for some computable function~$f$. 

An important tool in the development of parameterized algorithms is \emph{problem kernelization}, which is a polynomial-time preprocessing by \emph{data reduction rules} yielding a \emph{problem kernel}. Herein, the goal is, given any problem instance~$I$ with parameter $k$, to produce an \emph{equivalent} instance~$I'$ with parameter $k'$ in polynomial time such that $k' \leq k$ and the size of $I'$ is bounded from above by some function $g$ only depending on $k$. The function $g$ is called the \emph{kernel size}. If $g$ is a polynomial, we say that the problem has a \emph{polynomial kernel}. The \emph{equivalence} of the instances is defined as follows: $(I,k)$ is a yes-instance if and only if $(I',k')$ is a yes instance. A reduction rule is \emph{safe} if it produces an equivalent instance. An instance is \emph{reduced} with respect to a set of data reduction rules if each of the data reduction rules has been exhaustively applied.

Some parameterized problems that are fixed-parameter tractable are unlikely to admit a polynomial kernel~\cite{Cyg+15,DF13}. Precisely, these problems do not admit a polynomial kernel unless \badstuffhappens. By using \emph{polynomial parameter transformations} we can transfer these kernel lower bounds to other problems~\cite{BodlaenderJK14, BodlaenderTY11}. A polynomial parameter transformation maps any instance $(I,k)$ of some parameterized problem $L$ in polynomial time to an equivalent instance $(I',k')$ of a parameterized problem $L'$ such that $k' \leq p(k)$ for some polynomial $p$.

\section{On Problem Kernelizations}
\label{sec:kernel}

We now discuss problem kernelizations for \textsc{STC} parameterized by~$k$ and $\ell$. First, we give a $4k$-vertex kernel and an $\mathcal{O}(\ell \cdot 2^{\ell})$-vertex kernel. Then, we show that there is no polynomial problem kernel for~$\ell$ unless \badstuffhappens.
An important concept for our kernelizations are \textit{weak cuts} which are defined as follows.

\begin{definition} \label{definition: Weak Cut}
Let~$L=(S_L, W_L)$ be an STC-labeling for a graph~$G=(V,E)$. A \emph{weak cut} for~$G$ under~$L$ is a cut~$C$ such that~$E_C \subseteq W_L$.
\end{definition}

\begin{proposition}
\label{proposition: Weak Cut}
Let~$L=(S_L, W_L)$ be an STC-labeling for a graph~$G=(V,E)$. If there is a weak cut $C=(V_1,V_2)$ with cut-set $E_C$, then there is an STC-labeling $L'=(S_{L'}, W_{L'})$ for $G'=(V, E \setminus E_C)$ such that $|S_{L'}| = |S_L|$.
\end{proposition}

\begin{proof}
  Define $L':=(S_L,W_L \setminus
  E_C)$.
  Obviously, $|S_{L'}| = |S_L|$ holds. It remains to show that $L'$ still satisfies the STC
  property. Assume there are edges $\lbrace u, v \rbrace, \lbrace v, w \rbrace \in S_{L}$ such that
  $\lbrace u,w \rbrace \not \in E \setminus E_C$. Since $L$ satisfies STC, we have
  $\lbrace u,w \rbrace \in E_C$. Without loss of generality we assume that $u \in V_1$
  and $w \in V_2$. The fact that there is still a path from $u$ to $v$ in
  $E \setminus E_C$ contradicts the property that $E_C$ is a cut-set. \qed
\end{proof}

\subsection{A $4k$-Vertex Kernel for Strong Triadic Closure} \label{4k Kernel Section}
We now show that \textsc{STC} parameterized by $k$ admits a kernel with at most $4 k$ vertices. In the kernelization, we make use of the concepts of \textit{critical cliques} and \textit{critical clique graphs} as introduced in \cite{Protti2009}. These concepts were also used for a kernelization for \textsc{Cluster Editing} which directly gives a $4k$-vertex kernel for \textsc{CD}, even though this is not claimed explicitly~\cite{Guo09}. 

\begin{definition}
A \emph{critical clique} in a graph $G=(V,E)$ is a clique $K$ where the vertices of $K$ all have the same neighbors in $V \setminus K$, and $K$ is maximal under this property.
\end{definition}

Obviously, for a given graph~$G=(V,E)$ there exists a partition~$\mathcal{K}$ of the vertex set~$V$ such that every~$K \in \mathcal{K}$ is a critical clique in~$G$. The critical clique graph is then defined as follows.

\begin{definition}
Given a graph $G=(V,E)$, let $\mathcal{K}$ be the collection of its critical cliques. The \emph{critical clique graph} $\mathcal{C}$ of~$G$ is the graph $(\mathcal{K},E_\mathcal{C})$ with
$\lbrace K_i, K_j \rbrace \in E_{\mathcal{C}} \Leftrightarrow \forall u \in K_i, v \in K_j: \lbrace u , v \rbrace \in E$ \text{.}
\end{definition}

For a critical clique $K$ we let $\mathcal{N}(K) := \bigcup_{K' \in N_\mathcal{C}(K)} K'$ denote the union of its neighbor cliques in the critical clique graph and $\mathcal{N}^2 (K) := \bigcup_{K' \in N_\mathcal{C}^2(K)} K'$ denote the union of the critical cliques at distance exactly two from~$K$.
The critical clique graph can be constructed in $\mathcal{O}(n+m)$ time~\cite{HsuL}. Note that the edges within a critical clique~$K$ are not part of any $P_3$. It is known that this kind of edges is labeled as \textit{strong} in every optimal solution for \text{STC} \cite{sintosL}.

In the following, we will distinguish between two types of critical cliques. We say that a critical clique~$K$ is \emph{open} if $\mathcal{N}(K)$ does not form a clique in $G$, and that~$K$~is \emph{closed} if $\mathcal{N}(K)$ forms a clique in $G$.
We will see that the number of vertices in open critical cliques is bounded for every yes-instance of STC. The main step of the kernelization is to delete large closed critical cliques. Before we give the concrete rule we provide two useful properties of closed critical cliques.
\begin{proposition} 
\label{proposition: Type 2 Cliques not connected}
If $K_1$ and $K_2$ are closed critical cliques, then $\lbrace K_1, K_2 \rbrace \not \in E_\mathcal{C}$.
\end{proposition}  

\begin{proof}
  Assume there is an edge $\lbrace K_1, K_2 \rbrace \in E_\mathcal{C}$. The inclusions
  $K_1 \subseteq \mathcal{N}(K_2)$ and $K_2 \subseteq \mathcal{N}(K_1)$ obviously hold.

  First, we consider the case
  \begin{align*}
    \mathcal{N}(K_1) \setminus K_2 = \mathcal{N}(K_2) \setminus K_1 \text{.}
  \end{align*}
  In this case, all vertices in $K_1 \cup K_2$ have the same neighbors in
  $V\setminus (K_1 \cup K_2)$, which is a contradiction to the maximality of $K_1$ and $K_2$
  since $K_1 \cup K_2$ forms a bigger critical clique.

  Second, we consider the case
  \begin{align*}
    \mathcal{N}(K_1) \setminus K_2 \neq \mathcal{N}(K_2) \setminus K_1 \text{.}
  \end{align*}
  Without loss of generality, assume that there exists a vertex
  $v \in \mathcal{N}(K_1) \setminus K_2$ with $v \not \in \mathcal{N}(K_2) \setminus
  K_1$.
  Then, for any $w \in K_2 $, the vertices $v$ and $w$ are contained in $\mathcal{N}(K_1)$ but
  not adjacent in $G$. This is a contradiction to the fact that $K_1$ is closed.~\qed
\end{proof}

The following proposition has been proven already for both types of critical cliques~\cite{konstantinidis_et_alL}. Since we only need it for closed critical cliques we provide here a simple proof for those.

\begin{proposition}
  \label{proposition: Neighbours of critical cliques}
Let $K$ be a closed critical clique, $v \in \mathcal{N}(K)$ and let $L=(S_L,W_L)$ be an optimal STC-labeling for $G$. Then $E(\lbrace v \rbrace, K) \subseteq S_L$ or $E(\lbrace v \rbrace, K) \subseteq W_L$.
\end{proposition}
\begin{proof}
  Assume~$L$ has a weak edge $\lbrace v, w_1 \rbrace \in E(\lbrace v \rbrace, K)$ and a strong
  edge $\lbrace v, w_2 \rbrace \in E(\lbrace v \rbrace, K)$. We define the following labeling
  $L^+ := ( S_L \cup \lbrace \lbrace v, w_1 \rbrace \rbrace, W_L \setminus \lbrace \lbrace v,
  w_1 \rbrace \rbrace)$
  for $G$. We prove that $L^+$ satisfies STC by showing that
  there is no strong $P_3$ containing $\lbrace v, w_1 \rbrace$. Let $e$ be an edge that shares
  exactly one endpoint with $\lbrace v, w_1 \rbrace$.
  
  \textbf{Case 1:} $e \in E(K \cup \mathcal{N}(K))$. Since $K$ is closed, $K \cup \mathcal{N}(K)$ forms a clique. Then~$\lbrace v, w_1 \rbrace$ and $e$ are edges between the vertices of a clique, so they do not form an induced $P_3$.
  
  \textbf{Case 2:} $e \in E(\mathcal{N}(K),\mathcal{N}^2(K))$. Since $w_1 \in K$ it holds that $e=\{v,u\}$ for some $u \in \mathcal{N}^2(K)$. Then $e$ is weak under
  $L^+$. Otherwise, $e$ and $\lbrace v, w_2 \rbrace$ form a strong $P_3$ under $L$ which
  contradicts the fact that $L$ is an STC-labeling. Hence,  $\lbrace v, w_1 \rbrace$ and $e$ do not
  form a strong $P_3$.

  The fact that $L^+$ is an STC-labeling for $G$ with $|W_L|-1$ weak edges contradicts the fact
  that $L$ is an optimal STC-labeling for $G$. \qed

\end{proof}

 Proposition~\ref{proposition: Neighbours of critical cliques} tells us, that for a closed critical cliques~$K$ the vertices in~$\mathcal{N}(K)$ are either strong neighbors or weak neighbors of every~$v \in K$. We will use this fact to conclude that, whenever a closed critical clique~$K$ is larger than the number of edges between its first and second neighborhood, $K \cup \mathcal{N}(K)$ forms a strong clique under an optimal STC-labeling. Moreover, we show that there is no strong edge connecting~$K \cup \mathcal{N}(K)$ with the rest of the graph. We formulate the reduction rule as follows.

\begin{Rule} \label{4k Kernel RedRule}
If~$G$ has a closed critical clique $K$ with $|K| > |E_G(\mathcal{N}(K),\mathcal{N}^2(K))|$, then remove $K$ and $\mathcal{N}(K)$ from $G$ and decrease $k$ by $|E_G(\mathcal{N}(K),\mathcal{N}^2(K))|$.
\end{Rule}

\begin{proposition}
Rule \ref{4k Kernel RedRule} is safe.
\end{proposition}
 
\begin{proof}
Let $K$ be a closed critical clique in $G$ with $|K| > |E_G(\mathcal{N}(K),\mathcal{N}^2(K))|$ and let~$G'$ be the reduced graph after deleting $K$ and $\mathcal{N}(K)$ from $G$. We show that~$G$ has an STC-labeling $L=(S_L,W_L)$ with $|W_L| \leq k$ if and only if~$G'$ has an STC-labeling $L'=(S_{L'},W_{L'})$ with $|W_{L'}| \leq k - |E_G(\mathcal{N}(K),\mathcal{N}^2(K))|$.

First, let $L'=(S_{L'},W_{L'})$ be an STC-labeling for $G'$ such that $|W_{L'}| \leq k - |E_G(\mathcal{N}(K),\mathcal{N}^2(K))|$. We define a labeling $L=(S_L,W_L)$ with $|W_L| \leq k$ for $G$ by setting
\begin{align*}
S_L := S_{L'} \cup E_G(K \cup \mathcal{N}(K)) \text{ and } W_L := W_{L'} \cup E_G(\mathcal{N}(K),\mathcal{N}^2(K)) \text{.}
\end{align*}
It remains to show that $L$ is an STC-labeling. Since the edges in $S_{L'}$ do not have a common endpoint with the edges in $E_G(K \cup \mathcal{N}(K))$ it suffices to show that there is no induced $P_3$ containing two edges $e_1, e_2 \in S_{L'}$ or~$e_1, e_2 \in E_G(K \cup \mathcal{N}(K))$. If~$e_1, e_2 \in S_{L'}$, then the edges do not form a strong $P_3$ since $L'$ is an STC-labeling. If $e_1, e_2 \in E_G(K \cup \mathcal{N}(K))$, then $e_1$ and $e_2$ are edges between vertices of a clique since $K$ is a closed critical clique. Hence, $e_1$ and~$e_2$ do not form an induced $P_3$. Since there is no strong $P_3$ under $L$, it follows that $L$ is an STC-labeling with~$|W_L| \leq k$.

Conversely, let $L=(S_L,W_L)$ be an optimal STC-labeling for $G$ with~$|W_L| \leq k$. We prove the safeness by using Proposition \ref{proposition: Weak Cut}. To this end, we show that $C = (K \cup \mathcal{N}(K), V \setminus (K \cup \mathcal{N}(K)))$ is a weak cut under~$L$.

Assume towards a contradiction that there is a vertex $v \in \mathcal{N}(K)$ that has a strong neighbor $w \in \mathcal{N}^2(K)$. Then, for each $u \in K$, the edge $\lbrace u, v \rbrace$ is weak under~$L$. Otherwise, $\lbrace u, v \rbrace$ and $\lbrace v, w \rbrace$ would form a strong $P_3$, which contradicts the fact that $L$ is an STC-labeling. Then, we have exactly $|K|$ weak edges in $E_G(\lbrace v \rbrace, K)$ and at most $|E_G(\mathcal{N}(K),\mathcal{N}^2(K))|$ strong edges in $E_G(\lbrace v \rbrace, \mathcal{N}^2(K))$. We define a new labeling $L^+ = (S_{L^+},W_{L^+})$ by
\begin{align*}
S_{L^+} &:= S_L \cup E_G(\lbrace v \rbrace, K) \setminus  E_G(\lbrace v \rbrace, \mathcal{N}^2(K)) \text{,}\\
W_{L^+} &:= W_L \cup E_G(\lbrace v \rbrace, \mathcal{N}^2(K)) \setminus E_G(\lbrace v \rbrace, K) \text{.}
\end{align*}
From $|V(K)|>|E_G(\mathcal{N}(K),\mathcal{N}^2(K))|$ we get that $|W_{L^+}| < |W_L|$. It remains to show that $L^+$ is an STC-labeling, which contradicts the fact that $L$ is an optimal STC-labeling.

Since we add edges $\lbrace u, v \rbrace$ with $u \in K$ to $S_{L^+}$ we need to show that no such edge is part of a strong $P_3$ under $L^+$. Let $(\lbrace u, v \rbrace, e)$ with $u \in K$ be a pair of edges that share exactly one endpoint. Consider the case $e \in E(\mathcal{N}(K),\mathcal{N}^2(K))$. It follows that $e \in W_L^+$ by the construction of $L^+$. Hence, $\lbrace u,v \rbrace$ and $e$ do not form a strong $P_3$ under $L^+$. Otherwise, $e \in E(K \cup \mathcal{N}(K))$. Then~$\{u,v\}$ and~$e$ do not form an induced~$P_3$ since~$K\cup \mathcal{N}(K)$ is a clique by the definition of closed critical cliques.

Since there is no strong $P_3$ under $L^+$, it follows that $L^+$ is an STC-labeling. In combination with the fact that $|W_{L^+}| < |W_L|$, we conclude that $L^+$ is an STC-labeling for $G$ with fewer weak edges than $L$ which contradicts the fact that $L$ is an optimal STC-labeling. This contradiction implies that there is no vertex in~$\mathcal{N}(K)$ that has a strong neighbor in~$\mathcal{N}^2(K)$ under~$L$. Consequently, $C = (K \cup \mathcal{N}(K), V \setminus (K \cup \mathcal{N}(K)))$ is a weak cut under $L$. By using Proposition \ref{proposition: Weak Cut}, we conclude that there exists an STC-labeling $L'=(S_{L'},W_{L'})$ with $|W_{L'}| \leq k - |E_G(\mathcal{N}(K),\mathcal{N}^2(K))|$ in $G'$, proving the safeness of Rule~1. \qed

\end{proof}

\begin{proposition} \label{proposition exhaustive runningtime 4k kernel}
Rule \ref{4k Kernel RedRule} can be applied exhaustively in $\mathcal{O}(n \cdot m)$ time.
\end{proposition}

\begin{proof} We describe how to apply Rule \ref{4k Kernel RedRule} exhaustively in four steps.

\textbf{Step 1:} As a first step, we compute the critical clique graph $\mathcal{G}$ of $G$ and save the size $|K|$ for each critical clique $K$. This can be done in $\mathcal{O}(n+m)$ time~\cite{HsuL}.

\textbf{Step 2:} Next, we mark the closed critical cliques. To this end, we define \emph{deficit values} $d^{\{u,v\}}_u$ and $d^{\{u,v\}}_v$ for every edge $\{u,v\} \in E$ by 
\begin{align*}
d^{\{u,v\}}_u := |\{w \in V \setminus \{u\} \mid \{v,w\} \in E \text{ and } \{u,w\} \not \in E \}| \text{.}
\end{align*}
The deficit values can be computed in $\mathcal{O}(n \cdot m)$ time. Observe that a critical clique $K \in \mathcal{K}$ is closed if and only if for every $v \in K$ and $u \in \mathcal{N}(K)$ it holds that $d^{\{u,v\}}_u=0$.

%

\textbf{Step 3:} Now, we mark those closed critical cliques $K$ that satisfy $|K| > |E_G(\mathcal{N}(K),\mathcal{N}^2(K))|$. We compute the size of $\mathcal{N}(K)$ by $|\mathcal{N}(K)|=\deg(v)-|K|+1$ for some arbitrary $v \in K$. The value of $|E_G(\mathcal{N}(K),\mathcal{N}^2(K))|$ can be computed by evaluating the sum
\begin{align*}
\sum_{u \in \mathcal{N}(K)} (\deg(u) - |K|-|\mathcal{N}(K)|+1) \text{.}
\end{align*}
This can be done in $\mathcal{O}(n)$ time.

\textbf{Step 4:} Finally, we apply Rule \ref{4k Kernel RedRule} on every critical clique that was identified in the previous steps. Afterwards, we update the deficit values in the following way. Let $e$ be an edge that was deleted by the application of Rule \ref{4k Kernel RedRule}. Then, at most one of the endpoints of $e$ remains in the graph after the application. Let $v$ be this endpoint. We update the deficit values $d^{\{u,v\}}_u$ for every $u \in N(v)$ that was not deleted. For a fixed edge $e$, this can be done in $\mathcal{O}(n)$ time.

Afterwards, we re-identify the closed critical cliques. For each critical clique~$K$ whose vertices are incident with a deleted edge $e$, we need to check whether it satisfies $d^{\{u,v\}}_u=0$ for every $v \in K$ and $u \in \mathcal{N}(K)$ and mark it as closed. If two marked cliques $K_1$ and $K_2$ are adjacent in $\mathcal{G}$, we merge them to a critical clique consisting of the vertices in $K_1 \cup K_2$, since closed critical cliques are not adjacent in $\mathcal{G}$ by Proposition \ref{proposition: Type 2 Cliques not connected}. This can be done in $\mathcal{O}(n)$ time.

We repeat Steps 3 and 4 until no more application of Rule \ref{4k Kernel RedRule} is possible. Since we can delete at most $m$ edges, the overall running time of those steps is~$\mathcal{O}(n \cdot m)$. \qed
\end{proof}

\begin{theorem}
    \label{thm:4k}
    \textsc{STC} admits a~$4k$-vertex-kernel, which can be computed in $\mathcal{O}(n \cdot m)$-time.
\end{theorem}

\begin{proof} From Proposition \ref{proposition exhaustive runningtime 4k kernel} we know that we can produce a reduced \textsc{STC} instance regarding Rule \ref{4k Kernel RedRule} in $\mathcal{O}(n \cdot m)$-time. It remains to show that the number of vertices in a reduced instance is at most $4k$ or it is a no-instance for \textsc{STC}. Let $(G=(V,E),k)$ be a reduced \textsc{STC} instance. We first show that the
  overall number of vertices in open critical cliques is bounded by $2k$. Let $K$ be an open
  critical clique. Since~$\mathcal{N}(K)$ does not form a clique in $G$ by
  definition, there are two vertices $u, w \in \mathcal{N}(K)$ with~$\lbrace u , w \rbrace \not \in E$. So, for every vertex $v \in K$, the edges~$\lbrace u, v \rbrace$ and $\lbrace v, w \rbrace$ form an induced $P_3$. It follows that each vertex
  in any open critical clique must have at least one weak neighbor. Consequently, if the
  overall number of vertices in open critical cliques is bigger than $2k$, there must be
  more than $k$ weak edges in any STC-labeling.

  Now, let $\mathbb{K}$ denote the set of all vertices in closed critical cliques and let $L=(S_L, W_L)$ be an optimal STC-labeling. We prove that~$|\mathbb{K}| \leq 2k$ if $|W_L| \leq k$. Intuitively, we show that there is a correspondence between the weak edges of~$L$ and all vertices in~$\mathbb{K}$ such that for every weak edge under~$L$ there are at most two distinct vertices in~$\mathbb{K}$. Formally, we give a mapping~$\Phi:~\mathbb{K}~\rightarrow~W_L$ such that for each $e \in W_L$ we have
  $|\Phi ^{-1}(e)| \leq 2$, where
  $\Phi^{-1} (e) := \lbrace v \in \mathbb{K} \mid \Phi(v)=e \rbrace \subseteq
  \mathbb{K}$. If $|W_L| \leq k$, this
  implies 
  \begin{align*}
    |\mathbb{K}| \leq \sum_{e \in W_L} |\Phi^{-1}(e)| \leq k \cdot 2.
  \end{align*}

  First, consider those closed critical cliques whose vertices have at least one weak
  neighbor under $L$. By Proposition \ref{proposition: Neighbours of critical cliques}, every
  weak neighbor of such a clique has weak edges to every vertex of the clique. So each
  vertex $w$ in such a clique is the endpoint of some weak edge $e \in W_L$ and we define~$\Phi(w) :=
  e$. 

  Second, consider those closed critical cliques whose vertices have only strong neighbors
  under~$L$. Since $G$ is a reduced graph, for each such clique $K$ it holds that
  $|E(\mathcal{N}(K), \mathcal{N}^2(K))| \geq |K|$. Thus, for each $v$ in such a clique we can
  define a set $\Lambda_v :=\lbrace \lbrace v, w \rbrace, \lbrace w, u \rbrace \rbrace$ with
  $w \in \mathcal{N}(K)$ and $u \in \mathcal{N}^2 (K)$ such that
  $\Lambda_{v_1} \cap \Lambda_{v_2} = \emptyset$ if $v_1$ and $v_2$ are different vertices of
  the same critical clique $K$.

  Observe that $v_1 \neq v_2$ implies $\Lambda_{v_1} \neq \Lambda_{v_2}$: If $v_1$ and $v_2$
  lie in the same clique $K$ it follows directly from the definition of $\Lambda_{v_1}$ and
  $\Lambda_{v_2}$. So let $v_1 \in K_1$ and $v_2 \in K_2$ with $K_1 \neq K_2$. Assume that
  $\Lambda_{v_1} = \Lambda_{v_2} = \lbrace \lbrace v_1, w \rbrace, \lbrace w, v_2 \rbrace
  \rbrace$.
  By Proposition~\ref{proposition: Type 2 Cliques not connected}, $E(K_1,K_2) =
  \emptyset$.
  Then, the edges $\lbrace v_1, w \rbrace$ and $\lbrace w, v_2 \rbrace$ form a strong $P_3$
  under $L$ which contradicts the fact that $L$ is an STC-labeling.

  Obviously every $\Lambda_v$ forms an induced $P_3$, so at least one $e_v \in \Lambda_v$ is weak
  under~$L$. We define $\Phi(v):=e_v$. Since $\Phi$ is now defined for each
  $v \in \mathbb{K}$, it remains to check that $|\Phi^{-1}(e)| \leq 2$ for every $e \in W_L$.

  \textbf{Case 1:} Let $e=\lbrace w, u \rbrace \in W_L$ such that one endpoint $w$ lies in
  $\mathbb{K}$. By Proposition~\ref{proposition: Type 2 Cliques not connected}, we get that
  $u \not \in \mathbb{K}$. Moreover, there is at most one~$v \in \mathbb{K}$ such that
  $e \in \Lambda_{v}$: Assume there are two distinct vertices $v_1, v_2 \in \mathbb{K}$ such
  that $\{w,u\} \in \Lambda_{v_1} \cap \Lambda_{v_2}$. By~$\Lambda_{v_1} \cap \Lambda_{v_2} \neq \emptyset$ we know from the definition of those sets,
  that $v_1$ and $v_2$ lie in two different closed critical cliques $K_1$ and $K_2$ and have only strong neighbors under $L$. Since $w \in \mathbb{K}$ and vertices
  from different closed critical cliques are not adjacent, we conclude that $v_1$ and $v_2$
  are both adjacent to $u$. Since $\lbrace v_1, v_2 \rbrace \not \in E$, the edges
  $\lbrace v_1, u \rbrace$ and $\lbrace u, v_2 \rbrace$ form a strong $P_3$ under $L$, which
  contradicts the fact that $L$ satisfies STC. Since $e$ has at most one endpoint in
  $\mathbb{K}$ and lies in at most one $\Lambda_{v}$, we conclude
  that~$|\Phi^{-1}(e)| \leq 2$.

  \textbf{Case 2:} Let $e=\lbrace w, u \rbrace \in W_L$ such that $u, w \not \in
  \mathbb{K}$.
  We show that~$e$ lies in at most two sets $\Lambda_{v_1}, \Lambda_{v_2}$. Assume
  $\lbrace w,u \rbrace \in \Lambda_{v_1} \cap \Lambda_{v_2} \cap \Lambda_{v_3}$. From
  $\Lambda_{v_1} \cap \Lambda_{v_2} \cap \Lambda_{v_3} \neq \emptyset$ we know that $v_1$,
  $v_2$, and $v_3$ lie in three different cliques $K_1$, $K_2$, and~$K_3$, which are closed critical
  cliques and have only strong neighbors under $L$. Without loss of generality assume
  that all vertices of $K_1$~and~$K_2$ are adjacent to $u$. From Proposition \ref{proposition: Type 2 Cliques not connected} we know that
  $E(K_1,K_2)= \emptyset$. Then, the edges $\lbrace v_1, u \rbrace$ and
  $\lbrace v_2, u \rbrace$ form a strong $P_3$ under $L$ which contradicts the fact that $L$
  satisfies the STC property. Since $e$ has no endpoint in $\mathbb{K}$ and lies in at most
  two $\Lambda_v$, we conclude that ~$|\Phi^{-1}(e)| \leq 2$.

  Since for all $e \in W_L$ it holds that $|\Phi^{-1}(e)| \leq 2$, it follows that
  $|\mathbb{K}| \leq 2k$ as described above. Hence,~$G$ has at most $2k + 2k = 4k$ vertices
  or there is no STC-labeling with at most~$k$ weak edges. $\qed$
\end{proof}

\subsection{An $\mathcal{O}(\ell \cdot 2^\ell)$-Vertex Kernel for Strong Triadic Closure} \label{l2l Kernel Section}
We show that \textsc{STC} parameterized by $\ell:=|E|-k$ admits a kernel with $\mathcal{O}(\ell \cdot 2^\ell)$ vertices. Let $G=(V,E)$ be a graph and let $M \subseteq E$ be a maximum matching in $G$. Note that, if~$|M| \geq \ell$, then~$G$ has an STC-labeling~$L=(M,E\setminus M)$ with~$|M|\geq \ell$ strong edges. Hence, we may assume that the size of a maximum matching in~$G$ is smaller than~$\ell$.  The intuitive idea behind our kernelization is to delete vertices from the independent set that is formed by the vertices that are not incident with any edge in~$M$. We partition the vertices of $G$ into 
\begin{enumerate}
\item[•] $V_M := \lbrace v \in V \mid \text{ } v \text{ is an endpoint of some } e \in M \rbrace$,
\item[•] $I_2 := \lbrace v \in V \setminus V_M \mid \exists \lbrace u,w \rbrace \in M: u \text{ and } w \text{ are both neighbors of }v  \rbrace$, and
\item[•] $I_1 := V \setminus (I_2 \cup V_M)$.
\end{enumerate}
Note that since $M$ is maximal, $I_1 \cup I_2$ is an independent set. We will see that the number of vertices in $I_2$ is upper-bounded by $\ell$ in every STC instance. The main step of the kernelization is to delete superfluous vertices form $I_1$.

We will say that two vertices $v_1, v_2 \in I_1$ \emph{are members of the same family} $F$ if $N(v_1) = N(v_2)$. Given a family $F$, we will refer to the neighborhood of the vertices in $F$ as $N(F) := N(v)$ for some $v \in F$.
\begin{Rule} \label{l2l Kernel RedRule}
For every family $F$ of vertices in $I_1$: If $|F|>|N(F)|$ then delete $|F|-|N(F)|$ of the vertices in $F$ and decrease $k$ by $(|F|-|N(F)|) \cdot |N(F)|$.
\end{Rule}

Note that Rule~\ref{l2l Kernel RedRule} decreases the value of~$k$ by the number of deleted edges. Hence, the value of the parameter $\ell$ does not change. 

\begin{proposition} \label{Poposition l2l Rule safe}
Rule~\ref{l2l Kernel RedRule} is safe.
\end{proposition}

\begin{proof}
Let $G'$ be the reduced graph after applying Rule \ref{l2l Kernel RedRule}. We  prove that there is an STC-labeling~$L'=(S_{L'},W_{L'})$ for $G'$ with $|S_{L'}| \geq \ell$ if and only if there is an STC-labeling $L=(S_L,W_L)$ with~$|S_L| \geq \ell$ for $G$.

Let $L'=(S_{L'},W_{L'})$ be an STC-labeling for $G'$ with $|S_{L'}| \geq \ell$. It is easy to see that we can define an STC-labeling $L$ for $G$ from $L'$ by adding the edges that were deleted during the reduction to $W_{L'}$.

Now, let $L=(S_L,W_L)$ be an STC-labeling for $G$ with $|S_L| \geq \ell$. Let $F = \lbrace u_1,\ldots,u_{|F|} \rbrace$ be a family of vertices in $I_1$ such that $|F|>|N(F)|$ and $N(F) := \lbrace v_1 , v_2 , \ldots , v_{|N(F)|} \rbrace \subseteq V_M$. Then, since $I_1$ is an independent set, for every $j=1, \ldots ,|N(F)|$ there is at most one $i=1,\ldots,|F|$ such that $\lbrace u_i , v_j \rbrace \in S_L$. Otherwise, $\lbrace u_{i_1}, v_{j} \rbrace$ and $\lbrace  u_{i_2}, v_{j} \rbrace$ form a strong $P_3$. It follows that there are at least $|F|-|N(F)|$ nodes in $F$ that have only weak neighbors. Consequently, each of these nodes $u$ forms a weak cut $(\lbrace u \rbrace, V\setminus \lbrace u \rbrace )$ under $L$. By applying Proposition~\ref{proposition: Weak Cut} on each of those $|F|-|N(F)|$ weak cuts, we conclude that there is an STC-labeling~$L'$ with $|S_{L'}| \geq \ell$ for $G'$. \qed
\end{proof}
 
\begin{theorem}
\textsc{STC} admits a kernel with $\mathcal{O}(\ell \cdot 2^{\ell})$ vertices that can be computed in time~$\Oh(\sqrt{n}m)$.
\end{theorem}
\begin{proof}
 By Proposition \ref{Poposition l2l Rule safe}, Rule \ref{l2l Kernel
    RedRule} is safe. It remains to show that the number of vertices in a
  reduced STC instance is $\mathcal{O}(\ell \cdot 2^{\ell})$.

  Let $(G=(V,E),k)$ be a reduced instance. Recall that $V_M$ is the set of all vertices that are
  endpoints of some $e \in M$. As argued above, any instance with~$|M|\ge \ell$ is a yes-instance, and hence we assume~$|M|<\ell$ in the following.  Therefore,~$V_M$ has less than~$2\ell$ vertices.

  Recall that $I_2$ is the set of vertices that are adjacent to both endpoints of some
  edge $\lbrace u , w \rbrace~\in~M$.  For each edge $\lbrace u, w \rbrace \in M$ we will find
  at most one vertex in~$I_2$ that is adjacent to $u$ and $w$. Otherwise, if there were two
  such vertices $v$ and~$v'$, we could define a bigger matching by
  $M^+ := M \cup \lbrace \lbrace u,v \rbrace, \lbrace w,v' \rbrace \rbrace \setminus \lbrace
  \lbrace u,w \rbrace \rbrace$
  which is a contradiction to the property that $M$ is a maximum matching. Thus, from~$|M| < \ell$ we conclude~$|I_2| \leq \ell$.

  It remains to show that the size of $I_1$ is upper-bounded after applying Rule~\ref{l2l
    Kernel RedRule}. We start with the following observation:
  \begin{Observation}
    Every edge $\lbrace u, w \rbrace \in M$ has at most one endpoint with neighbors in~$I_1$.
  \end{Observation}
  \begin{proof}
  Assume that there are edges $\lbrace u, v_1 \rbrace$ and $\lbrace w, v_2 \rbrace$ for some
  $v_1, v_2 \in I_1$. By the fact that $v_1 \in I_1$ and therefore $v_1 \not \in I_2$ it holds that $v_1~\neq~v_2$. We define a matching
  $M^+ := M \cup \lbrace \lbrace u,v_1 \rbrace, \lbrace w,v_2 \rbrace \rbrace \setminus \lbrace
  \lbrace u,w \rbrace \rbrace$
  that is bigger than $M$, which contradicts the fact that $M$ is a maximum
  matching. $\hfill \Diamond$
  \end{proof}

  Also, there is no edge from some vertex in $I_1$ to some vertex in $I_2$, since
  $I_1 \cup I_2$ is an independent set. Since $|M|< \ell$, there are less than
  $\ell$ vertices with neighbors in $I_1$. Thus, there are less than $2^\ell$ different
  families $F$ of vertices in $I_1$. Since $G$ is a reduced graph with respect to Rule \ref{l2l
    Kernel RedRule}, the size of each family is at most $\ell$. Hence,
  $|I_1| \leq \ell \cdot 2^{\ell}$, which delivers a problem kernel with
  $\mathcal{O}(\ell \cdot 2^\ell)$ vertices.

  Finally, the running time can be seen as follows. Computing a maximum matching can be done
  in~$\Oh(\sqrt{n}m)$ time~\cite{MV80L} and all other steps including the exhaustive application
  of Rule~\ref{l2l Kernel RedRule} can be performed in linear time.  $\qed$ \end{proof}

\iflong If we do not distinguish between~$I_1$ and $I_2$, we can compute a problem kernel of size $\mathcal{O}(\ell \cdot 4^{\ell})$ in linear time: In this case, we only need~to compute a maximal matching instead of a maximum matching which leads to $2^{2\cdot \ell}$ different families of vertices in $I_1 \cup I_2$, increasing the kernel size to $\mathcal{O}(\ell \cdot 4^{\ell})$.
\fi

\subsection{A Kernel Lower Bound for  the Parameter $\ell$} \label{No Poly Kernel Section}

Above, we gave an exponential-size problem kernel for \textsc{STC} parameterized by the number of strong edges~$\ell$. Now we prove that \textsc{STC} does not admit a polynomial kernel for the parameter $\ell$ unless \badstuffhappens{} by reducing from~\textsc{Clique}.
\begin{center}
	\begin{minipage}[c]{.9\linewidth}
		\textsc{Clique}
		
		\textbf{Input}: $G=(V,E)$, $t \in \mathbb{N}$
		\\
		\textbf{Question}: Is there a clique on $t$ vertices in $G$? 
	\end{minipage}
\end{center}

%

\textsc{Clique} parameterized by the size $s$ of a vertex cover does not admit a polynomial kernel unless~\badstuffhappens~\cite{BodlaenderJK14}. Our proof gives a polynomial parameter transformation~\cite{BodlaenderTY11} from \textsc{Clique} parameterized by $s$ to \textsc{STC} parameterized by $\ell$ in two steps. The first step is a reduction to the following~problem.
\begin{center}
	\begin{minipage}[c]{.9\linewidth}
		\textsc{Restricted Multicolored Clique}
		
		\textbf{Input}: A properly $t$-colored graph $G=(V,E)$ with color classes $C_1,\ldots,C_t \subseteq V$ such that~$|C_1|=|C_2|= \ldots =|C_{t-1}|$. 
		\\
		\textbf{Question}: Is there a clique containing one vertex from each color in $G$? 
	\end{minipage}
\end{center}

\begin{proposition}
  \label{Clique_by_VC reduced to MulitClique_by_VC}
\textsc{Restricted Multicolored Clique} parameterized by $|C_1 \cup \ldots \cup C_{t-1}|$ does not admit a polynomial kernel unless \badstuffhappens.
\end{proposition}

\begin{proof} We give a polynomial parameter transformation from \textsc{Clique} parameterized
  by the size $s$ of a vertex cover to \textsc{Restricted Multicolored Clique} parameterized by
  $|C_1 \cup \ldots \cup C_{t-1}|$.

  Let $(G,t)$ be a \textsc{Clique} instance with a size-$s$ vertex cover
  $S=\lbrace v_1,\ldots,v_{s} \rbrace$ and let $I := V \setminus S$ be the remaining
  independent set. Since $I$ is an independent set, the maximal value for the clique
  size $t$ is $s+1$. Otherwise $(G,t)$ is a trivial no-instance.  We construct an instance
  $G'$ for \textsc{Restricted-Multicolored-Clique} as follows.

  First, we define $t$ color classes $C_1,\ldots,C_t$. We replace every vertex $v_i \in S$ with
  $t$ copies $v_{i,1},\ldots, v_{i,t}$ such that
  $v_{i,1} \in C_1, v_{i,2} \in C_2,\ldots, v_{i,t} \in C_t$. We also add all vertices of $I$
  to $C_t$. Now, the classes $C_1,\ldots,C_{t-1}$ contain exactly $s$ elements and the class
  $C_t$ contains $s+|I|$ elements. If two vertices $v_i$ and $v_j$ are adjacent in $S$, we
  connect all copies of $v_i$ with all copies of $v_j$, except those that are in the same color class. For the first $(t-1)$ classes $C_1,\ldots, C_{t-1}$ we do the following: For
  every edge $\lbrace v_i,w \rbrace$ with $v_i \in S$ and $w \in I$, we add edges
  $\lbrace v_{i,j}, w \rbrace$ for each copy $v_{i,j}$ of $v_i$. Since $C_1,\ldots,C_t$ are all
  independent sets, $G'$ is a properly $t$-colored graph with $t-1$ color classes of the same
  size $s$. Hence, $G'$ is a feasible instance for
  \textsc{Restricted-Multicolored-Clique}. Note that $t \leq s+1$ implies
  $|C_1 \cup \ldots \cup C_{t-1}| = (t-1) \cdot s \leq s^2$. To prove that the transformation
  from $(G,t)$ into $(G',C_1,C_2,\ldots,C_t)$ is a polynomial parameter transformation, it
  remains to show that $G'$ has a multicolored clique if and only if $G$ has a clique of size
  $t$.

  Let $K$ be a clique of size $t$ in $G$. Since $I$ is an independent set we can assume that
  $t-1$ vertices of $K$ lie in $S$ and one vertex $u$ of $K$ lies in $S \cup I$. Without loss of
  generality, we assume that $K=\lbrace v_1,\ldots, v_{t-1},u \rbrace$. If $u \in S$, it holds
  that $u=v_i$ for some $i \geq t$. Then there is a copy $v_{i,t} \in C_t$ of $v_i$ and the
  vertices $v_{1,1},\ldots,v_{t-1,t-1}, v_{i,t}$ form a multicolored clique in $G'$. If
  $u\in I$, it follows that $u \in C_t$ and the vertices $v_{1,1},\ldots,v_{t-1,t-1},u$ form a
  multicolored clique in $G'$.

  Now let $G'$ have a multicolored clique
  $\lbrace v_{i_1,1}, v_{i_2,2}, \ldots , v_{i_{(t-1)},t-1}, u \rbrace$ with $u \in C_t$. Note
  that the indices $i_1,\ldots,i_{(t-1)}$ are pairwise distinct by construction. If $u \in I$,
  the vertices $v_{i_1}, v_{i_2}, \ldots , v_{i_{(t-1)}}, u$ form a clique of size $t$ in
  $G$. Otherwise, if $u \not \in I$, we can assume that $u=v_{i_t,t}$ for some $i_t$ that is
  different from $i_1,\ldots,i_{t-1}$. Then, the vertices
  $v_{i_1}, v_{i_2}, \ldots , v_{i_{(t-1)}}, v_{i_t}$ form a clique of size $t$ in $G$. \qed
\end{proof}



The next step to prove the kernel lower bound is to give a polynomial parameter transformation from \textsc{Restricted-Multicolored-Clique} to \textsc{STC}.
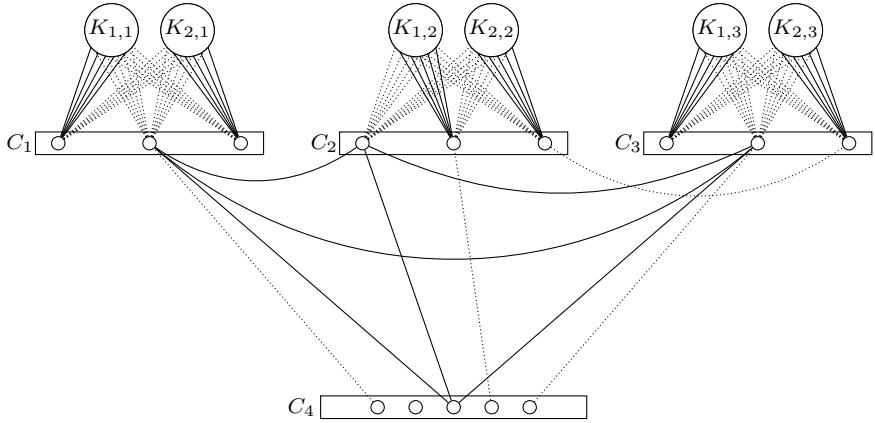
\begin{figure}[t]
\begin{center}
\begin{tikzpicture}
\tikzstyle{knoten}=[circle,fill=white,draw=black,minimum size=5pt,inner sep=0pt]
\tikzstyle{clique}=[circle,fill=white,draw=black,minimum size=20pt,inner sep=0pt]

\node[clique] (K11) at (0,1) {$K_{1,1}$};
\node[clique] (K21) at (1,1) {$K_{2,1}$};

\draw (0.5,-0.5) node[minimum height=0.3cm,minimum width=3cm,draw] {};
\node at (-1.2,-0.5) {$C_1$};
\node[knoten] (v11) at (-0.7,-0.5) {};
\node[knoten] (v21) at (0.5,-0.5) {};
\node[knoten] (v31) at (1.7,-0.5) {};

\foreach \x in {-0.2,-0.1,...,0.2}{
          \draw[-] (v11) to ($(K11)+(1.3*\x,-0.2)$);
          \draw[densely dotted] (v11) to ($(K21)+(1.3*\x,-0.2)$);
          \draw[densely dotted] (v21) to ($(K11)+(1.3*\x,-0.2)$);
          \draw[densely dotted] (v21) to ($(K21)+(1.3*\x,-0.2)$);
          \draw[densely dotted] (v31) to ($(K11)+(1.3*\x,-0.2)$);
          \draw[-] (v31) to ($(K21)+(1.3*\x,-0.2)$);
      }
      
\node[clique] (K12) at (4,1) {$K_{1,2}$};
\node[clique] (K22) at (5,1) {$K_{2,2}$};

\draw (4.5,-0.5) node[minimum height=0.3cm,minimum width=3cm,draw] {};
\node at (2.8,-0.5) {$C_2$};
\node[knoten] (v12) at (3.3,-0.5) {};
\node[knoten] (v22) at (4.5,-0.5) {};
\node[knoten] (v32) at (5.7,-0.5) {};

\foreach \x in {-0.2,-0.1,...,0.2}{
          \draw[densely dotted] (v12) to ($(K12)+(1.3*\x,-0.2)$);
          \draw[densely dotted] (v12) to ($(K22)+(1.3*\x,-0.2)$);
          \draw[-] (v22) to ($(K12)+(1.3*\x,-0.2)$);
          \draw[densely dotted] (v22) to ($(K22)+(1.3*\x,-0.2)$);
          \draw[densely dotted] (v32) to ($(K12)+(1.3*\x,-0.2)$);
          \draw[-] (v32) to ($(K22)+(1.3*\x,-0.2)$);
      }
      
\node[clique] (K13) at (8,1) {$K_{1,3}$};
\node[clique] (K23) at (9,1) {$K_{2,3}$};

\draw (8.5,-0.5) node[minimum height=0.3cm,minimum width=3cm,draw] {};
\node at (6.8,-0.5) {$C_3$};
\node[knoten] (v13) at (7.3,-0.5) {};
\node[knoten] (v23) at (8.5,-0.5) {};
\node[knoten] (v33) at (9.7,-0.5) {};

\foreach \x in {-0.2,-0.1,...,0.2}{
          \draw[-] (v13) to ($(K13)+(1.3*\x,-0.2)$);
          \draw[densely dotted] (v13) to ($(K23)+(1.3*\x,-0.2)$);
          \draw[densely dotted] (v23) to ($(K13)+(1.3*\x,-0.2)$);
          \draw[densely dotted] (v23) to ($(K23)+(1.3*\x,-0.2)$);
          \draw[densely dotted] (v33) to ($(K13)+(1.3*\x,-0.2)$);
          \draw[-] (v33) to ($(K23)+(1.3*\x,-0.2)$);
      }

\draw (4.5,-4) node[minimum height=0.3cm,minimum width=3.5cm,draw] {};
\node at (2.5,-4) {$C_4$};
\node[knoten] (v14) at (3.5,-4) {};
\node[knoten] (v24) at (4,-4) {};
\node[knoten] (v34) at (4.5,-4) {};
\node[knoten] (v44) at (5,-4) {};
\node[knoten] (v64) at (5.5,-4) {};

\draw[-] (v21) to (v34);
\draw[-] (v23) to (v34);
\draw[-] (v12) to (v34);
\draw[-] (v21) edge[-,bend right=35] (v12);
\draw[-] (v21) edge[-,bend right=40] (v23);
\draw[-] (v12) edge[-,bend right=25] (v23);
\draw[densely dotted] (v44) to (v22);
\draw[densely dotted] (v14) to (v21);
\draw[densely dotted] (v64) to (v23);
\draw[densely dotted] (v32) edge[-,bend right=35] (v33);


\node[clique] (K11) at (0,1) {$K_{1,1}$};
\node[clique] (K21) at (1,1) {$K_{2,1}$};

\node[clique] (K12) at (4,1) {$K_{1,2}$};
\node[clique] (K22) at (5,1) {$K_{2,2}$};

\node[clique] (K13) at (8,1) {$K_{1,3}$};
\node[clique] (K23) at (9,1) {$K_{2,3}$};

\end{tikzpicture}
\end{center}
\caption{An example for the construction of $G'$ described in the proof of Theorem \ref{thm:no-kernel-stc} with $t=4$ color classes. The top of the picture shows color classes $C_1$, $C_2$, and~$C_3$ of size $z=3$ and their attached cliques. The bottom shows color class $C_4$. The edges between the color classes are the edges from $G$. The dotted edges correspond to the weak edges of an optimal STC-labeling for $G'$.}
\label{Figure Construction NoPolyKernel}
\end{figure} 
\begin{theorem}
  \label{thm:no-kernel-stc}
  \textsc{STC} parameterized by the number of strong edges~$\ell$ does not admit a polynomial
  kernel unless \badstuffhappens.
\end{theorem}
\begin{proof}
We give a polynomial parameter transformation from \textsc{Restricted Multicolored Clique} to \textsc{STC}. Let $G=(V,E)$ be a properly $t$-colored graph with color classes $C_1=\lbrace v_{1,1}, v_{2,1},\ldots, v_{z,1} \rbrace$, $C_2=\lbrace v_{1,2}, v_{2,2},\ldots, v_{z,2} \rbrace$, \ldots, $C_{t-1}=\lbrace v_{1,t-1}, v_{2,t-1},\ldots, v_{z,t-1} \rbrace$, each of size $z$, and $C_t$.
We now describe how to construct an STC instance $(G'=(V',E'),k)$ from $G$ such that there is an STC-labeling $L=(S_L,W_L)$ with $|W_L| \leq k$ for $G'$ if and only if $G$ has a multicolored~clique.

For each of the first $(t-1)$ classes $C_r$, $r=1,\ldots,t-1$, we define a family $\mathcal{K}_r$ of $z-1$ vertex sets $K_{1,r}, K_{2,r},\ldots,K_{z-1,r}$, each of size $t$, and we add edges such that each $K \in \mathcal{K}_r$ becomes a clique. Throughout this proof, those vertex sets will be called \emph{attached cliques}. For every fixed $i=1,\ldots,z-1$ we also add edges $\lbrace u, v \rbrace$ from all $u \in K_{i,r}$ to all $v \in C_r$. Fig.~\ref{Figure Construction NoPolyKernel} shows an example of this construction.

Setting
 $k := |E| - \left(\binom{t}{2}+(t-1)(z-1)\binom{t+1}{2}\right)$ 
 gives us $\ell =\binom{t}{2}+(t-~1)\cdot(z-~1)\binom{t+1}{2}$. Obviously, $\ell$ is polynomially bounded in $|C_1 \cup \ldots \cup C_{t-1}| = (t-1) \cdot z$.
 
 We now prove that the construction of $(G',k)$ from $(G,C_1,C_2,\dots,C_t)$ is a correct polynomial parameter
  transformation, which means that there is an STC-labeling $L=(S_L,W_L)$ with $|W_L| \leq k$
  (or equivalently $|S_L| \geq \ell$) for $G'$ if and only if $G$ has a multicolored clique.

  $(\Rightarrow)$ Let $M$ be a multicolored clique in $G'$. Without loss of generality we can assume that
  $M = \lbrace v_{z,1},\ldots, v_{z,t-1}, u \rbrace$ with $u \in C_t$. Consider the following
  labeling~$L=(S_L,W_L)$ for $G'$. We set $S_L := E_{M} \cup E_{\mathcal{K}} \cup E_{C}$ to be
  the disjoint union of the following edge sets:

\begin{align*}
&E_{M} := E(M) \text{,}\\
&E_{\mathcal{K}} := \bigcup_{\substack{i=1,\ldots,z-1\\r=1,\ldots,t-1}} E(K_{i,r}) \text{, and}\\
&E_{C} := \bigcup_{\substack{i=1,\ldots,z-1\\r=1,\ldots,t-1}} E(\lbrace v_{i,r} \rbrace, K_{i,r}) \text{.}
\end{align*}
The set $E_M$ is the set of all edges between vertices of $M$, $E_{\mathcal{K}}$ contains all edges between the vertices of the attached cliques, and $E_C$ contains all edges between vertices $v_{i,r}$, $i<z$ and the vertices in the corresponding attached clique $K_{i,r}$.

It remains to show that $L := (S_L, E' \setminus S_L)$ is an STC-labeling with $|S_L| \geq \binom{t}{2}+(t-1)(z-1)\binom{t+1}{2}$.
The size of $S_L$ is easy to check:
\begin{align*}
|S_L| &= |E_{M}| + |E_{\mathcal{K}}| +|E_{C}| \\
	  &= \binom{t}{2}+(t-1)(z-1)\binom{t}{2} + (t-1)(z-1)t \\
	  &= \binom{t}{2}+(t-1)(z-1)\binom{t+1}{2} \text{.}
\end{align*}
Hence, we need to check that there is no strong $P_3$ under $L$. Let $e_1, e_2 \in S_L$. 

\textbf{Case 1:} $e_1, e_2 \in E_{M}$ or $e_1, e_2 \in E_{\mathcal{K}}$:\\
In this case, all endpoints of $e_1$ and $e_2$ lie in the same clique in $G'$, so $e_1$ and $e_2$ do not~form~an induced~$P_3$.

\textbf{Case 2:} $e_1, e_2 \in E_{C}$:\\
If $e_1$ and $e_2$ share exactly one endpoint, they have the form $e_1= \lbrace v_{i,r} , w_1 \rbrace$, $e_2 = \lbrace v_{i,r}, w_2 \rbrace$ with $w_1, w_2 \in K_{i,r}$ for some $i=1,\ldots,z-1$ and $r=1,\ldots,t-1$ by the definition of $E_{C}$. Then, there exists an edge $\lbrace w_1, w_2 \rbrace$ by the construction of $G'$, so $e_1$ and $e_2$ do not form an induced $P_3$.

\textbf{Case 3:} ($e_1 \in E_{M}$ and $e_2 \in E_{\mathcal{K}}$) or ($e_1 \in E_{M}$ and $e_2 \in E_{C}$):\\
In this case, $e_1$ and $e_2$ do not share an endpoint by the construction of $G'$, so they do not form an induced $P_3$.

\textbf{Case 4:} $e_1 \in E_{\mathcal{K}}$ and $e_2 \in E_{C}$:\\
If $e_1$ and $e_2$ share exactly one endpoint, they have the form $e_1= \lbrace v_{i,r} , w_1 \rbrace$, $e_2 = \lbrace w_1, w_2 \rbrace$ with $w_1, w_2 \in K_{i,r}$ for some $i=1,\ldots,z-1$ and $r=1,\ldots,t-1$. Then, there exists an edge $\lbrace v_{i,r}, w_2 \rbrace$ by the construction of $G'$, so $e_1$ and $e_2$ do not form an induced $P_3$.

We have thus shown that there is no strong $P_3$ under $L$. Hence, the labeling $L$ is an STC-labeling with $|S_L| = \binom{t}{2}+(t-1)(z-1)\binom{t+1}{2}$.

$(\Leftarrow)$ Conversely, assume that there is an STC-labeling $L=(S_L,W_L)$ with a maximal number of strong edges for $G'$ such that $|S_L| \geq  \binom{t}{2}+(t-1)(z-1)\binom{t+1}{2}$. Consider a color class~$C_r$ for fixed $r=1,\ldots,t-1$. Let $v \in C_r$ be some vertex in~$C_r$. We can make the following important observations that follow directly from the construction of~$G'$.

\begin{Observation} \label{Observation1}
Every pair $(\lbrace v,u \rbrace, \lbrace v, w \rbrace)$ of edges with $u \in K$ for some $K \in \mathcal{K}_r$ and $w~\in~C_{r'},~r'~\neq~r$, forms a $P_3$. Hence,~$\{u,v\}$ or~$\{v,w\}$ is weak under~$L$.\end{Observation}

\begin{Observation} \label{Observation2}
For fixed $v \in C_r$, there are at most $t-1$  strong edges under $L$ of the form $\lbrace v, w \rbrace$ with $w~\in~C_{r'},$ $r'~\neq~r$. If~$v$ has exactly~$t-1$ strong neighbors, $v$ has exactly one strong neighbor in each other color class.\end{Observation}

\begin{Observation} \label{Observation3}
For fixed $v \in C_r$, there are up to~$t$ strong edges under $L$ of the form~$\lbrace v, w \rbrace$, $w \in \bigcup_{K \in \mathcal{K}_r} K$. If there are exactly~$t$ strong edges of such form, $v$ forms a strong $(t+1)$-clique with one of the attached cliques $K \in \mathcal{K}_r$.\end{Observation}


Now, consider an attached clique $K \in \mathcal{K}_r$ for some $r=1,\ldots,t-1$. Note that $K$ is a critical clique as defined in Section \ref{4k Kernel Section}. Hence, all edges between vertices in $K$ are strong under $L$. Without loss of generality, we can make the following assumption for $K$.

\begin{Assumption} \label{Assumption1}
For fixed $K \in \mathcal{K}_r$, there is at most one $j=1,\ldots,z$ such that there are strong edges $\lbrace w, v_{j,r} \rbrace \in E(K,C_r)$ under $L$.
\end{Assumption}

\begin{proof} [of Assumption \ref{Assumption1}]
If there was another strong edge $\lbrace w', v_{j',r} \rbrace \in E(K_{i,r},C_r)$ with $j' \neq j$, then $w' \neq w$ as otherwise $\lbrace w', v_{j',r} \rbrace$ and $\lbrace w, v_{j,r} \rbrace$ form a strong $P_3$ under $L$. We can thus define an STC-labeling $L^+ :=(S_{L^+},W_{L^+})$ with $|S_{L^+}|=|S_L|$ by 
\begin{align*}
S_{L^+} := S_L \cup \lbrace \lbrace w', v_{j,r} \rbrace \rbrace \setminus \lbrace \lbrace w', v_{j',r} \rbrace \rbrace \text{.}
\end{align*}
It is easy to check that $L^+$ still satisfies STC. $\hfill \Diamond$
\end{proof}

Assumption \ref{Assumption1} leads to the following two observations.

\begin{Observation} \label{Observation4}
If some $v \in C_r$ has a strong neighbor under $L$ in one clique $K \in \mathcal{K}_r$ it holds that $E(\lbrace v \rbrace, K) \subseteq S_L$.
\end{Observation}

\begin{proof} [of Observation \ref{Observation4}]
If there were weak and strong edges under $L$ in $E(\lbrace v \rbrace, K)$, we could define a new STC-labeling $L^+$ by adding all of $E(\lbrace v \rbrace, K)$ to $S_L$. Then, $L^+$ has more strong edges than $L$. Because of Assumption \ref{Assumption1}, $v$ is the only strong neighbor of the vertices in $K$ under $L^+$, and because of Observation \ref{Observation1}, the vertices of $K$ are the only strong neighbors of $v$ under $L^+$. Hence, $L^+$ satisfies STC. This contradicts the fact that $L$ is an STC-labeling with a maximal number of strong edges.$\hfill \Diamond$
\end{proof}

\begin{Observation} \label{Observation5}
There are $z-1$ vertices in $C_r$ that form a strong clique under $L$ of size $t+1$ with one of the attached cliques $K \in \mathcal{K}_r$.
\end{Observation}

\begin{proof} [of Observation \ref{Observation5}]
Assume that there are at least two vertices $v, w \in C_r$ that do not form a strong clique with some $K \in \mathcal{K}_r$. From Observation \ref{Observation4}, we conclude that $v$ and $w$ do not have any strong neighbor in $\bigcup_{K \in \mathcal{K}_r} K$. From Assumption \ref{Assumption1} we conclude that there is at least one $K \in \mathcal{K}_r$ such that the vertices in $K$ have no strong neighbors in $C_r$. Then, we can define a new labeling $L^+=(S_{L^+}, W_{L^+})$ by
\begin{align*}
S_{L^+} &:= S_L \cup E( \lbrace v \rbrace, K) \setminus E( \lbrace v \rbrace, \bigcup_{\substack{j=1,\ldots, t\\j \neq r}} C_j) \text{.}
\end{align*}
From Observation \ref{Observation2} and Observation \ref{Observation3} we conclude that $|S_{L^+}|>|S_L|$. Since every vertex in $K$ has only weak neighbors in $C_r \setminus \lbrace v \rbrace$ and $v$ has only weak neighbors in~$V' \setminus K$ under~$L^+$, the labeling $L^+$ satisfies STC, which contradicts the fact that~$L$ is an STC-labeling with a maximal number of strong edges. $\hfill \Diamond$
\end{proof}

From Observation \ref{Observation5} we know that there are $(t-1) \cdot (z-1) \cdot \binom{t+1}{2}$ strong edges under $L$ in those strong cliques of size $t+1$. Since $|S_L| \geq  \binom{t}{2}+(t-1)(z-1)\binom{t+1}{2}$, there are at least $\binom{t}{2}$ further edges that are strong under $L$. In the following, we describe how we can find a multicolored clique in $G'$ using these $\binom{t}{2}$ strong edges.

Let $R := \lbrace v_{1,1} ,v_{1,2} ,\ldots, v_{1,t-1} \rbrace \subseteq C_1 \cup \ldots \cup C_{t-1}$ be the set of vertices that do not form a strong clique with any attached $K$. Since $C_t$ is an independent set, each $v \in R$ has at most one strong neighbor in $C_t$. Hence, there are at most $t-1$ strong edges in $E(R, C_t)$. Since $\binom{t}{2} - (t-1) = \binom{t-1}{2}$, there must be $\binom{t-1}{2}$ strong edges between the vertices in $R$. Hence, $G'[R]$ is a complete subgraph, and $R$ is a strong clique under $L$.

Now let $U \subseteq C_t$ be the subset of vertices in $C_t$ that have a strong neighbor in $R$. The set $U$ is not empty since $|S_L| \geq  \binom{t}{2}+(t-1)(z-1)\binom{t+1}{2}$. Each $u \in U$ must have edges to each $v \in R$. Otherwise, if $\lbrace u ,v \rbrace \in E'$ and $\lbrace u ,w \rbrace \not \in E'$ for some $v,w \in R$, the edges $\lbrace u ,v \rbrace$ and $\lbrace w ,v \rbrace$ form a strong $P_3$ under $L$. Hence $R \cup \lbrace u \rbrace$ is a multicolored clique in $G'$ which proves the correctness of the reduction. \qed
\end{proof}

The proof of Theorem~\ref{thm:no-kernel-stc}  also implies that \textsc{CD} has no kernel with respect to the parameter~$\ell:=|E|-k$: The strong
edges in the STC-labeling obtained in the forward direction of the proof form a disjoint
union of cliques and the converse direction follows from the fact that a cluster subgraph
with at least~$\ell$ cluster edges implies an STC-labeling with at least~$\ell$ strong
edges which then implies that the \textsc{Multicolored Clique} instance is a yes-instance.
\begin{corollary}\label{cor:no-kernel-cd}
  \textsc{CD} parameterized by the number of cluster edges $\ell:=|E|-k$
  does not admit a polynomial kernel unless \badstuffhappens.
\end{corollary}

\section{Fixed-Parameter Algorithms for the Parameterization by the Number of Strong Edges or Cluster Edges}\label{sec:l-fpt}
For \textsc{CD}, we obtain a fixed-parameter algorithm by a simple dynamic programming algorithm.
\begin{theorem}
  \label{thm:l-cd}
\textsc{CD} can be solved in~$\Oh(9^{\ell}\cdot \ell n)$ time.
\end{theorem}

\begin{proof}
  The first step of the algorithm is to compute a maximal matching~$M$ in~$G$.
  If~$|M|\ge \ell$, then answer yes. Otherwise, since~$M$ is maximal, the
  endpoints of~$M$ are a vertex cover of size less than~$2\ell$. Let~$C$ denote this vertex cover and
  let~$I:= V\setminus C$ denote the independent set consisting of the vertices that are not an
  endpoint of~$M$. We now decide if there is a cluster subgraph with at least~$\ell$ cluster edges using dynamic
  programming over subsets of~$C$. Assume in the following that~$I:=\{1,\ldots, n-|C|\}$. The
  dynamic programming table~$T$ has entries of the type~$T[i,C']$ for
  all~$i\in \{0,1,\ldots, n-|C|\}$ and all~$C'\subseteq C$. Each entry stores the maximum
  number of cluster edges in a clustering of~$G[C'\cup \{1,\ldots ,i\}]$. After filling this
  table completely, we have a yes-instance if~$T[n-|C|,C]\ge \ell$ and a no-instance
  otherwise. The entries are computed for increasing values of~$i$ and subsets~$C'$ of
  increasing size. Note that the entry for~$i=0$ corresponds to the clusterings that contain no
  vertices of~$I$. The recurrence to compute an entry for~$i=0$ is

  $$T[0,C'] = \max_{C''\subseteq C': C'' \text{is a clique}} T[0,C'\setminus C''] + \binom{|C''|}{2}.$$

  The recurrence to compute an entry for~$i\ge 1$ is
 
  $$T[i,C'] = \max_{C''\subseteq C': C''\cup \{i\} \text{is a clique}} T[i-1,C'\setminus C''] + \binom{|C''|+1}{2}.$$

  The correctness follows from the observation that we consider all cases for the clique
  containing~$i$ since~$i$ is not adjacent to any vertex~$j<i$.

  The running time of the algorithm can be seen as follows. A maximal matching can be computed
  in linear time. If the matching has size less than~$\ell$, we fill the dynamic programming
  table as defined above. For each~$i$, the number of terms that are evaluated in the
  recurrences is~$3^{|C|}$ as each term corresponds to one partition of~$C$
  into~$C\setminus C'$,~$C'\setminus C''$, and~$C''$. For each term one needs to determine
  in~$\Oh(\ell^2)$ time whether~$C''\cup \{i\}$ is a clique. Hence, the overall time needed to
  fill~$T$ is~$\Oh(3^{2\ell}\cdot \ell n)=\Oh(9^{\ell}\cdot \ell n)$. $\qed$
\end{proof}

For \textsc{STC}, we combine a branching on the graph that is induced by a maximal
matching with a dynamic programming over the vertex sets of this graph.
\begin{theorem}\label{thm:l-stc}\textsc{STC} can be solved in~$\ell^{\Oh(\ell)}\cdot n$ time.
\end{theorem}
\begin{proof}
  The initial step of the algorithm is to compute a maximal matching~$M$ in~$G$.
  If~$|M|\ge \ell$, then answer yes. Otherwise, the
  endpoints of~$M$ are a vertex cover of size less than~$2\ell$ since~$M$ is maximal. Let~$C$ denote this vertex cover and
  let~$I:= V\setminus C$ denote the independent set consisting of the vertices that are
  not an endpoint of~$M$. The algorithm now has two further main steps. First, try all
  STC-labelings of~$G[C]$ with at most~$\ell$ strong edges. If there is one STC-labeling
  with~$\ell$ strong edges, then answer yes. Otherwise, compute for each STC-labeling
  of~$G[C]$ with fewer than~$\ell$ edges, whether it can be extended to an STC-labeling
  of~$G$ with~$\ell$ strong edges by labeling sufficiently many edges
  of~$E(C,I)$ as strong.

  Observe that~$G[C]$ has~$\ell^{\Oh({\ell})}$ STC-labelings with at most~$\ell$ strong
  edges and that they can be enumerated in~$\ell^{\Oh({\ell})}$ time: The graph~$G[C]$ has less than~$\binom{2\ell}{2}=\Oh(\ell^2)$ edges and we enumerate all subsets
  of size at most~$\ell$ of this set. Now consider one such set~$S_C$. In~$\Oh({\ell^2})$
  time, we can check whether~$(S_C,E(C)\setminus S_C)$ is a
  valid~STC-labeling. If this is not the case, then
  discard the current set. Otherwise, compute whether this labeling can be extended into a
  labeling of~$G$ with at least~$\ell$ strong edges by using dynamic programming over
  subsets of~$C$. Assume in the following that~$I:=\{1,\ldots, n-|C|\}$.  The dynamic
  programming table~$T$ has entries of the type~$T[i,C']$ for
  all~$i\in \{1,\ldots, n-|C|\}$ and all~$C'\subseteq C$. Each entry stores the maximum
  number of strong edges in an STC-labeling of~$G[C\cup \{1,\ldots ,i\}]$ in which the strong
  edges of~$E(C)$ are exactly those of~$S_C$ and in which the strong neighbors of the
  vertices in~$\{1,\ldots ,i\}$ are exactly from~$C'$. Observe that the set of strong
  neighbors~$N_S(i)$ of each vertex~$i$ has to fulfill three properties:
  \begin{itemize}
  \item $N_S(i)$ is a clique.
  \item No vertex of~$N_S(i)$ has a strong neighbor in~$C\setminus N(i)$.
  \item No vertex of~$N_S(i)$ has a strong neighbor in~$I \setminus \{i\}$.
  \end{itemize}
  We call a set that fulfills the first two properties \emph{valid} for~$i$. We ensure the
  third property by the recurrence in the dynamic programming.

 After filling this table completely, we have a yes-instance if~$T[n-|C|,C]\ge \ell$. Otherwise, the current STC-labeling for~$G[C]$ cannot be extended to an STC-labeling for~$G$ with at least~$\ell$ strong edges. If~$T[n-|C|,C]< \ell$ for every choice of an STC-labeling for~$G[C]$ we have a no-instance. The entries are computed for increasing values of~$i$ and
  subsets~$C'$ of increasing size. The basic entry is~$T[0,\emptyset]$ which is set
  to~$|S_C|$. The recurrence to compute an entry for~$i\ge 1$ is
   $$T[i,C'] = \max_{C''\subseteq C': C''  \text{is valid for}  \{i\}} T[i-1,C'\setminus C''] + {|C''|}.$$
   The correctness follows from the observation that we consider all valid sets for strong
   neighbors and that in the optimal solution for~$G[i-1,C'\setminus C'']$ no vertex
   from~$\{1, \ldots, i-1\}$ has strong neighbors in~$C''$.

   The running time of the algorithm can be seen as follows. A maximal matching can be
   computed greedily in linear time. If the matching has size less than~$\ell$, we fill the dynamic
   programming table as defined above. The number of partial labelings~$S_C$
   is~$\ell^{\Oh(\ell)}$. For each of them, in~$\Oh(2^{2\ell}\cdot \ell n)$ time, we can
   compute for each~$i$ the subsets of~$C$ which are valid for~$i$. The number of terms
   that are subsequently evaluated in the recurrences is~$3^{|C|}$ as each term
   corresponds to one partition of~$C$ into~$C\setminus C'$,~$C'\setminus C''$,
   and~$C''$. For each term, one needs to evaluate the equation in~$\Oh(1)$ time. Hence,
   the overall time needed to fill~$T$ for one partial labeling~$S_C$
   is~$\Oh(3^{2\ell}\cdot n)=\Oh(9^{\ell}\cdot n)$; the overall running time follows. $\qed$
\end{proof}

\section{Strong Triadic Closure and Cluster Deletion on $H$-free graphs}\label{sec:stc-cd}

Recall that every solution for \textsc{CD} provides an STC-labeling $L=(S_L,W_L)$ by defining $S_L$ as the set of edges inside the cliques in the resulting graph. We call such~$L$ a \textit{cluster labeling}. However, this labeling is not necessarily optimal~\cite{konstantinidis_et_alL}.

In this section we discuss the complexity and the solution structure if the input for \textsc{STC} and \textsc{CD} is limited to $H$-free graphs, that is, graphs that do not have an induced subgraph~$H$. We give a dichotomy for all classes of $H$-free graphs, where~$H$~is a graph on three or four vertices.
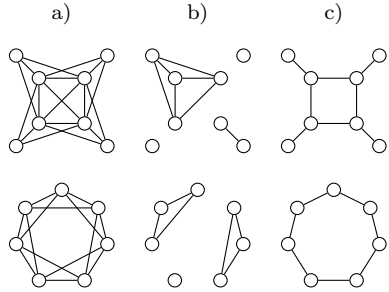
\begin{figure}[t]
\begin{center}
 \setlength{\tabcolsep}{2mm}
  \renewcommand{\arraystretch}{4}
\begin{tabular}{ c c c }
a)&b)&c)\\
	\begin{tikzpicture}[scale=0.6]
		\tikzstyle{knoten}=[circle,fill=white,draw=black,minimum size=5pt,inner sep=0pt]
		\node[knoten] (1) at (0,0) {};
		\node[knoten] (2) at (0,1) {};
		\node[knoten] (3) at (1,1) {};
		\node[knoten] (4) at (1,0) {};
		
		\node[knoten] (a) at (-0.5,1.5) {};
		\node[knoten] (b) at (1.5,1.5) {};
		\node[knoten] (c) at (-0.5,-0.5) {};
		\node[knoten] (d) at (1.5,-0.5) {};

		\draw[-] (1) edge (2);
		\draw[-] (1) edge (3);
		\draw[-] (1) edge (4);
		\draw[-] (2) edge (3);
		\draw[-] (2) edge (4);
		\draw[-] (3) edge (4);
		
		\draw[-] (a) edge (1);
		\draw[-] (a) edge (2);
		\draw[-] (a) edge (3);
		
		\draw[-] (b) edge (2);
		\draw[-] (b) edge (3);
		\draw[-] (b) edge (4);
		
		\draw[-] (c) edge (1);
		\draw[-] (c) edge (4);
		\draw[-] (c) edge (2);
		
		\draw[-] (d) edge (1);
		\draw[-] (d) edge (4);
		\draw[-] (d) edge (3);

	\end{tikzpicture}
	&
	\begin{tikzpicture}[scale=0.6]
		\tikzstyle{knoten}=[circle,fill=white,draw=black,minimum size=5pt,inner sep=0pt]
		\node[knoten] (1) at (0,0) {};
		\node[knoten] (2) at (0,1) {};
		\node[knoten] (3) at (1,1) {};
		\node[knoten] (4) at (1,0) {};
		
		\node[knoten] (a) at (-0.5,1.5) {};
		\node[knoten] (b) at (1.5,1.5) {};
		\node[knoten] (c) at (-0.5,-0.5) {};
		\node[knoten] (d) at (1.5,-0.5) {};

		\draw[-] (1) edge (2);
		\draw[-] (1) edge (3);
		\draw[-] (2) edge (3);
		
		\draw[-] (a) edge (1);
		\draw[-] (a) edge (2);
		\draw[-] (a) edge (3);
		
		\draw[-] (d) edge (4);
	\end{tikzpicture}
	&
	\begin{tikzpicture}[scale=0.6]
		\tikzstyle{knoten}=[circle,fill=white,draw=black,minimum size=5pt,inner sep=0pt]
		\node[knoten] (1) at (0,0) {};
		\node[knoten] (2) at (0,1) {};
		\node[knoten] (3) at (1,1) {};
		\node[knoten] (4) at (1,0) {};
		
		\node[knoten] (a) at (-0.5,1.5) {};
		\node[knoten] (b) at (1.5,1.5) {};
		\node[knoten] (c) at (-0.5,-0.5) {};
		\node[knoten] (d) at (1.5,-0.5) {};

		\draw[-] (1) edge (2);
		\draw[-] (2) edge (3);
		\draw[-] (3) edge (4);
		\draw[-] (4) edge (1);
		
		\draw[-] (a) edge (2);
		
		\draw[-] (b) edge (3);
		
		\draw[-] (c) edge (1);
		
		\draw[-] (d) edge (4);
	\end{tikzpicture}
	\\

	\begin{tikzpicture}[scale=0.6]
		\tikzstyle{knoten}=[circle,fill=white,draw=black,minimum size=5pt,inner sep=0pt]
		\node[knoten] (1) at (-1,0) {};
		\node[knoten] (2) at (-0.8,0.8) {};
		\node[knoten] (3) at (0,1.2) {};
		\node[knoten] (4) at (0.8,0.8) {};
		\node[knoten] (5) at (1,0) {};
		\node[knoten] (6) at (0.5,-0.8) {};
		\node[knoten] (7) at (-0.5,-0.8) {};
		
		\draw[-] (1) edge (2);
		\draw[-] (2) edge (3);
		\draw[-] (3) edge (4);
		\draw[-] (4) edge (5);
		\draw[-] (5) edge (6);
		\draw[-] (6) edge (7);
		\draw[-] (7) edge (1);
		
		\draw[-] (1) edge (3);
		\draw[-] (2) edge (4);
		\draw[-] (3) edge (5);
		\draw[-] (4) edge (6);
		\draw[-] (5) edge (7);
		\draw[-] (6) edge (1);
		\draw[-] (7) edge (2);
	\end{tikzpicture}
	&
	\begin{tikzpicture}[scale=0.6]
		\tikzstyle{knoten}=[circle,fill=white,draw=black,minimum size=5pt,inner sep=0pt]
		\node[knoten] (1) at (-1,0) {};
		\node[knoten] (2) at (-0.8,0.8) {};
		\node[knoten] (3) at (0,1.2) {};
		\node[knoten] (4) at (0.8,0.8) {};
		\node[knoten] (5) at (1,0) {};
		\node[knoten] (6) at (0.5,-0.8) {};
		\node[knoten] (7) at (-0.5,-0.8) {};
		
		\draw[-] (1) edge (2);
		\draw[-] (2) edge (3);
		\draw[-] (4) edge (5);
		\draw[-] (5) edge (6);
		
		\draw[-] (1) edge (3);
		\draw[-] (4) edge (6);
	\end{tikzpicture}
	&
	\begin{tikzpicture}[scale=0.6]
		\tikzstyle{knoten}=[circle,fill=white,draw=black,minimum size=5pt,inner sep=0pt]
		\node[knoten] (1) at (-1,0) {};
		\node[knoten] (2) at (-0.8,0.8) {};
		\node[knoten] (3) at (0,1.2) {};
		\node[knoten] (4) at (0.8,0.8) {};
		\node[knoten] (5) at (1,0) {};
		\node[knoten] (6) at (0.5,-0.8) {};
		\node[knoten] (7) at (-0.5,-0.8) {};
		
		\draw[-] (1) edge (2);
		\draw[-] (2) edge (3);
		\draw[-] (3) edge (4);
		\draw[-] (4) edge (5);
		\draw[-] (5) edge (6);
		\draw[-] (6) edge (7);
		\draw[-] (7) edge (1);
		
	\end{tikzpicture}
\\
\end{tabular}
\end{center}
\caption{Two graphs where no cluster labeling is an optimal STC-labeling. Column a) shows the input graph, column~b) shows an optimal cluster labeling, and column c) shows the strong edges in an optimal STC-labeling.}
\label{Figure Examples ClusterLabel}
\end{figure} 
\subsection{The Correspondence between Strong Triadic Closure and Cluster Deletion on~$H$-Free Graphs}
We say that the two problems \emph{correspond} on a graph class $\Pi$ if for every graph in~$\Pi$ we can find a cluster labeling that is an optimal STC-labeling. In this case we call the labeling an \emph{optimal cluster labeling}.

Fig.~\ref{Figure Examples ClusterLabel} shows two examples, where a cluster labeling is not an optimal solution for \textsc{STC}. The upper example, provided by Konstantinidis and Papadopoulos~\cite{konstantinidis_et_alL}, is $C_4$-, $2K_2$-, co-paw-, and co-diamond-free; an optimal STC-labeling has eight strong edges, while the best cluster labeling has only seven cluster edges. The second example is the complement of a $C_7$. It is $3K_1$-, $K_4$-, $4K_1$-, claw-, and co-claw free; the optimal STC-labeling has seven strong edges, while the best cluster labeling has six cluster edges. The examples give the cases where STC and \textsc{CD} do not correspond.

\begin{theorem}
  \label{STC=CD}
The problems \textsc{CD} and \textsc{STC}
\begin{enumerate}
\item[•] do not correspond on the class of $H$-free graphs, for $H \in \lbrace 3K_1, C_4, 2K_2, \text{co-paw},$\\
$\text{co-diamond},  K_4, 4K_1, \text{claw}, \text{co-claw} \rbrace$, and
\item[•] correspond on the class of $H$-free graphs, for $H \in \lbrace K_3, P_3, K_2 + K_1, P_4, \text{diamond},$ $\text{paw} \rbrace$. 
\end{enumerate}
\end{theorem}

\begin{proof} The examples in Fig.~\ref{Figure Examples ClusterLabel} show that \textsc{CD} and \textsc{STC} do not correspond on $H$-free graphs for $H \in \lbrace 3K_1, C_4, 2K_2, \text{co-paw}, \text{co-diamond},  K_4,$ $4K_1, \text{claw}, \text{co-claw} \rbrace$. It remains to show the correspondence for $H \in \lbrace K_3, P_3, K_2 + K_1, P_4, \text{diamond},$ $\text{paw} \rbrace$.

\textbf{Case 1:} $H=K_3$. On triangle-free graphs, the strong edges of an optimal STC-labeling correspond to the edges of a maximum matching, which is obviously an optimal cluster labeling.

\textbf{Case 2:} $H=P_3$. On $P_3$-free graphs, every edge is labeled as strong in an optimal STC-labeling \cite{sintosL}. Since $P_3$-free graphs are cluster graphs, this labeling clearly is a cluster labeling.

\textbf{Case 3:} $H =P_4$. There is an optimal cluster labeling on $P_4$-free~graphs \cite{KNP17L}.

\textbf{Case 4:} $H = \text{paw}$. It is known that every component of a paw-free graph is triangle-free or complete multipartite \cite{olariu1988paw}. Complete multipartite graphs are $P_4$-free, so it follows by the Cases~1 and~3 that there is an optimal cluster labeling on paw-free graphs.

\textbf{Case 5:} $H = K_2 + K_1$. Every $K_2 + K_1$-free graph is paw-free, so it follows by Case~4 that there is an optimal cluster labeling on $K_2 + K_1$-free graphs.

\textbf{Case 6:} $H = \text{diamond}$. 
Let $G=(V,E)$ be a diamond-free graph. We prove that there is an optimal cluster labeling for $G$. It is known that the class of diamond-free graphs can be characterized as \textit{strictly clique irreducible} graphs \cite{prisner1993}. A graph is called strictly clique irreducible if every edge in the graph lies in a unique maximal clique. 

To show that there is an optimal cluster labeling, it is sufficient to prove that there is an optimal STC-labeling $L=(S_L,W_L)$ such that there is no triangle~$u_1, u_2, u_3 \in V$ with~$\lbrace u_1, u_2 \rbrace, \lbrace u_2, u_3 \rbrace \in S_L$ and~$\lbrace u_1, u_3 \rbrace \in W_L$.

Let $v \in V$ be some vertex of $G$. Since $G$ is strictly clique irreducible, we can partition $N[v]$ into maximal cliques $K_1,K_2,\ldots,K_t$ such that $K_i \cap K_j = \lbrace v \rbrace$ for $i \neq j$. Let $L=(S_L,W_L)$ be an optimal STC-labeling for $G$ such that $v$ has a strong neighbor $w_1$ in $K_1$ under $L$. We prove that $v$ does not have a strong neighbor in any of the other maximal cliques, which means $E(\lbrace v \rbrace, N(v) \setminus K_1) \subseteq W_L$. Assuming $v$ has a strong neighbor $w_j \in K_j$ for some $j \neq 1$, there must be an edge $\lbrace w_1, w_j \rbrace \in E$ since $L$ satisfies STC. Then, there is a clique $K_+ \subseteq N[v]$ containing $v$, $w_1$, and $w_j$, which contradicts the fact that $G$ is strictly clique irreducible.

Now assume that there is a triangle $u_1, u_2, u_3 \in V$ such that $\lbrace u_1, u_2 \rbrace, \lbrace u_2, u_3 \rbrace \in S_L$ and $\lbrace u_1, u_3 \rbrace \in W_L$. Since every vertex can only have strong neighbors in one maximal clique, $u_1$, $u_2$, and $u_3$ are elements of the same maximal clique $K$. Since $u_1$ and $u_3$ do not have any strong neighbors in $V \setminus K$, we do not produce a strong~$P_3$ by adding $\lbrace u_1, u_3 \rbrace$ to $S_L$, which contradicts the fact that $L$ is an optimal STC-labeling. \qed
\end{proof}

\subsection{The Complexity of Strong Triadic Closure and Cluster Deletion on~$H$-Free~Graphs}

We first identify the cases where both problems are solvable in polynomial time.

\begin{lemma} \label{poly-time solvable}
If $H \in \lbrace K_3, P_3, K_2 + K_1, P_4, \text{paw} \rbrace$, \textsc{STC} and \textsc{CD} are solvable in polynomial time on $H$-free graphs.
\end{lemma}

\begin{proof}
\textsc{STC} and \textsc{CD} are solvable in polynomial time on $P_4$-free graphs \cite{KNP17L}. On triangle-free graphs, we can solve both problems by computing a maximal matching, which can be done in polynomial time~\cite{MV80L}. On $P_3$-free graphs we can find a trivial solution by labeling every edge strong. It is known that every component of a paw-free graph is triangle-free or complete multipartite and thus~$P_4$-free~\cite{olariu1988paw}. Hence, we can use the polynomial-time algorithm on the $P_4$-free components and a polynomial-time algorithm to find a  maximum matching on the triangle-free components. Since $K_2 + K_1$-free graphs are paw-free, we can use the same polynomial-time algorithm on these graphs. \qed
\end{proof}
In all other possible cases for $H$, both problems remain NP-hard on $H$-free graphs. For the case~$H \in \{3K_1, 4K_1, 2K_2, \text{claw},\text{co-diamond}, \text{co-paw}\}$ we give a reduction from \textsc{Clique} that has been previously used to show certain hardness results for~\textsc{CD}~\cite{N99}. For the case~$H = \text{co-claw}$ we provide a slightly more complicated reduction. The remaining cases~$H \in \{C_4, \text{diamond}, K_4\}$ follow from previous results~\cite{KU12}. Consider the following construction.

\begin{definition}
Let $G=(V,E)$ be a graph. The \emph{expanded graph} $\tilde{G}$ of $G$ is the graph obtained by 
  adding a clique $\tilde{K} = \lbrace v_1, \dots , v_{|V|^3} \rbrace$ and
  edges such that every $v \in V$ is adjacent to all vertices in $\tilde{K}$.
\end{definition}

Obviously, we can construct $\tilde{G}$ from $G$ in polynomial time. We use this construction to give a 
reduction from \textsc{Clique} to  \textsc{STC} and \textsc{CD}. The construction also transfers certain $H$-freeness properties from $G$ to $\tilde{G}$.
\begin{lemma} 
  \label{Clique Reduction STC and CD} Let $(G=(V,E),t)$ be a
    \textsc{Clique} instance.
    \begin{enumerate}
    \item[$(a)$] There is a clique of size at least $t$ in $G$ if and only if there is an
      STC-labeling $L=(S_L,W_L)$ for $\tilde{G}$ such that
      $|S_L| \geq \binom{n^3}{2} + t \cdot n^3$.
    \item[$(b)$] There is a clique of size at least $t$ in $G$ if and only if $\tilde{G}$ has a solution
      for \textsc{CD} with at least
      $\binom{n^3}{2} + t \cdot n^3$ cluster edges.
    \end{enumerate}
  \end{lemma}
  \begin{proof}
$(a)$ Let $V' \subseteq V$ be a clique on $t$ vertices in $G$. Then we obtain an STC-labeling~$L=(S_L,W_L)$ for~$\tilde{G}$ with at least~$\binom{n^3}{2} + t \cdot n^3$ strong edges by defining~$S_L := E(V' \cup \tilde{K})$. Note that $L$ is a cluster labeling, so it obviously satisfies~STC. From~$|\tilde{K}|=n^3$ we also get that $|S_L| = \binom{t}{2} + \binom{n^3}{2} + t \cdot n^3 \geq \binom{n^3}{2} + t \cdot n^3$.

Conversely, let there be an STC-labeling $L=(S_L,W_L)$ for $\tilde{G}$ with at least~$\binom{n^3}{2} + t \cdot n^3$ strong edges. We show that there is a clique $V'$ of size at least $t$ in $G$. Observe that, since $|E_{\tilde{G}}(\tilde{K})| = \binom{n^3}{2}$, there are at least $t \cdot n^3$ edges in~$S_L \setminus E_{\tilde{G}}(\tilde{K})$.

We will call two vertices $v_1, v_2 \in \tilde{K}$ \textit{member of the same family} $F$, if $v_1$ and $v_2$ have the exact same strong neighbors in $V$. For each family $F$, the set of strong neighbors of $F$ forms a clique. Otherwise, if there were two non-adjacent~strong neighbors $u$ and $w$ of any $v \in F$, the edges $\lbrace u, v \rbrace$ and $\lbrace v, w \rbrace$ form a strong $P_3$ under~$L$. Let $K_F$ denote the set of strong neighbors of a family~$F$.

Let $F_1, \dots , F_p \neq \emptyset$ be the families of vertices in $\tilde{K}$. It holds that 
\begin{align*}
|S_L \setminus (E_{\tilde{G}}(\tilde{K}) \cup E) | = \sum_{i=1}^{p} |F_i| \cdot |K_{F_i}| \leq \max_{i} |K_{F_i}| \cdot \sum_{i=1}^{p} |F_i| = \max_{i} |K_{F_i}| \cdot n^3 \text{.}
\end{align*}
Since $|E| \leq \binom{n}{2}$ it holds that $(t-1) \cdot n^3 < t \cdot n^3 - \binom{n}{2} \leq |S_L \setminus (E_{\tilde{G}}(\tilde{K}) \cup E) | $. We conclude that $(t-1) < \max_{i} |K_{F_i}|$. Hence, there is a clique of size at least $t$ in~$G$.

$(b)$ Let $V' \subseteq V$ be a clique on $t$ vertices in $G$. Since the labeling $L$ from the proof of Claim $(a)$ is a cluster labeling, there must be a solution for \textsc{CD} on $\tilde{G}$ with at least $\binom{n^3}{2} + t \cdot n^3$ cluster edges.

Now, let there be a solution for \textsc{CD} on $\tilde{G}$ such that there are at least $\binom{n^3}{2} + t \cdot n^3$ cluster edges. We define an STC-labeling $L=(S_L,W_L)$ for~$\tilde{G}$ by defining $S_L$ as the set of cluster edges in the solution. Then, there is an STC-labeling with at least $\binom{n^3}{2} + t \cdot n^3$ strong edges. It then follows by $(a)$, that $G$ has a clique of size at least $t$, which completes the proof. \qed \end{proof}

\begin{lemma} \label{Proposition Gtilde erhaelt freiheit}
Let $H \in \lbrace 3K_1, 2K_2, \text{co-diamond}, \text{co-paw}, 4 K_1 \rbrace$. If a graph $G$ is $H$-free, then the expanded graph $\tilde{G}$ is $H$-free as well.
\end{lemma}

\begin{proof}
Note that each $H \in \lbrace 3K_1, 2K_2, \text{co-diamond}, \text{co-paw}, 4 K_1 \rbrace$ is disconnected. Assume $\tilde{G}$ has $H$ as an induced subgraph. Since $G$ is $H$-free and we do not add any edges between vertices of $G$ during the construction of $\tilde{G}$, at least one of the vertices of this induced subgraph lies in $\tilde{K}$. By construction, each vertex in $\tilde{K}$ is adjacent to every other vertex in $\tilde{G}$, which contradicts the fact that the induced subgraph is disconnected. \qed
\end{proof}

We next use Lemmas \ref{Clique Reduction STC and CD} and \ref{Proposition Gtilde erhaelt freiheit} to obtain NP-hardness results for~\textsc{STC} and~\textsc{CD}. Note that in case of~$H \in \{3 K_1, 2 K_2\}$ the NP-hardness for \textsc{CD} is already known~\cite{GHN13}. Moreover, in case of~$H \in \{2K_2,C_4\}$, the NP-hardness of \textsc{STC} on~$H$-free graphs is implied by the NP-hardness of~\textsc{STC} on split graphs~\cite{konstantinidis_et_alL}.

\begin{lemma} \label{NP-hardness}
\textsc{STC} and \textsc{CD} remain NP-hard on $H$-free graphs if 
$$H \in \lbrace 3K_1, 2K_2, \text{co-diamond}, \text{co-paw}, 4 K_1, \text{claw}, C_4, \text{diamond}, K_4 \rbrace\text{.}$$
\end{lemma}

\begin{proof}
\textbf{Case 1:} $H \in\lbrace 3K_1, 2K_2, \text{co-diamond}, \text{co-paw}, 4 K_1 \rbrace$. \textsc{Clique} remains NP-hard on $3K_1$-,$2K_2$-, co-diamond-, co-paw- and $4 K_1$-free graphs since \textsc{Independent Set} is NP-hard on the complement graphs: $K_3$-, $C_4$-, diamond-, paw-, and $K_4$-free graphs~\cite{poljak1974}. By Lemma \ref{Clique Reduction STC and CD}, $(G,k) \mapsto (\tilde{G}, m - ( \binom{n^3}{2} + k \cdot n^3 ))$ is a polynomial-time reduction from \textsc{Clique} to \textsc{STC} and \textsc{CD}. From Lemma \ref{Proposition Gtilde erhaelt freiheit} we know that if $G$ is $ 3K_1$-, $2K_2$-, $\text{co-diamond}$-, $\text{co-paw}$- or $4 K_1$-free, so is $\tilde{G}$. Thus, \textsc{STC} and \textsc{CD} remain NP-hard on $H$-free graphs.

\textbf{Case 2:} $H = \text{claw}$. Since both problems are NP-hard on $3K_1$-free graphs due to Case 1, it follows, that they are NP-hard on claw-free graphs.


\textbf{Case 3:} $H \in\lbrace C_4, \text{diamond}, K_4 \rbrace$. There is a reduction from \textsc{3Sat} to \textsc{CD}  producing a $C_4$-, $K_4$-, and diamond-free  \textsc{CD} instance~\cite{KU12}. By Theorem \ref{STC=CD}, there is an optimal cluster labeling for \textsc{STC} on diamond-free graphs, so the reduction works also for~\textsc{STC}. Thus, \textsc{STC} and \textsc{CD} remain NP-hard on $C_4$-, $K_4$-, and diamond-free graphs. \qed
\end{proof}
It remains to show NP-hardness on co-claw-free graphs.
\iflong Since \textsc{Independent Set} can be solved in polynomial time on claw-free graphs \cite{bihi80}, we can solve \textsc{Clique} on co-claw-free graphs in polynomial time. Hence, \else Since \textsc{Clique} can be solved  in polynomial time on co-claw-free graphs~\cite{bihi80}
\fi we cannot reduce from \textsc{Clique} \iflong as in the proof of Lemma~\ref{NP-hardness}\else in this case\fi.
Instead, we reduce from the following~problem.

\begin{center}
	\begin{minipage}[c]{.9\linewidth}
		\textsc{3-Clique Cover}\\
		\textbf{Input}: A graph $G=(V,E)$
		\\
		\textbf{Question}: Can $V$ be partitioned into three cliques $K_1, K_2$, and $K_3$?
	\end{minipage}
\end{center}

\textsc{3-Clique Cover} is NP-hard on co-claw-free graphs\iflong, since \textsc{3-Colorability} is NP-hard on claw-free graphs~\cite{Holyer}.\fi

\begin{lemma}
  \label{co-claw-free hardness}
  \textsc{STC} and \textsc{CD} remain NP-hard on co-claw-free graphs.
\end{lemma}
\begin{proof}
We give a reduction from \textsc{3-Clique Cover} on co-claw-free graphs to \textsc{STC} and \textsc{CD} on co-claw free graphs.
Let $G=(V,E)$ be a co-claw-free instance for \textsc{3-Clique Cover}. We construct a co-claw-free \textsc{STC} instance $(G'=(V',E'),k)$, which is an equivalent \textsc{CD} instance, as follows.

We define three vertex sets $K_1$,$K_2$, and~$K_3$. Every $K_i$ consists of exactly $n^3$ vertices $v_{1,i}, \dots, v_{n^3,i}$. We set $V' := V \cup K_1 \cup K_2 \cup K_3$. Moreover, we define edges from every vertex in $K_1 \cup K_2 \cup K_3$ to every vertex in $V$ and edges of the form $\lbrace v_{c,i}, v_{d,j} \rbrace$, where $c \neq d$. The set $E'$ is the union of those edges and $E$. Note that this makes each $K_i$ a clique of size $n^3$. We set $k:= |E'| - (3 \cdot \binom{n^3}{2} + n^4)$\iflong, and let~$\mathcal{K} := K_1 \cup K_2 \cup K_3$ denote the union of these cliques\fi. It remains to prove the following three claims.
\begin{enumerate}[$(a)$]
\item $G'$ is co-claw-free.
\item $G \mapsto (G',k)$ is a correct reduction from \textsc{3-Clique Cover} to~\textsc{STC}.
\item $G \mapsto (G',k)$ is a correct reduction from \textsc{3-Clique Cover} to~\textsc{CD}.
\end{enumerate}

 $(a)$ A co-claw consists of a triangle and an isolated vertex. Assume that $G'$ has a co-claw as an induced subgraph. Then there is a triangle on some vertices $u_1, u_2, u_3 \in V'$ and a vertex $w$ such that $w$ has no edge to one of the $u_i$.

\textbf{Case 1:} $w \in \mathcal{K}$. Without loss of generality we assume $w= v_{p,1}$ for some $p=1, \dots, n^3$. By the construction of $G'$, the vertex $w$ has edges to every other vertex of $G'$ except $v_{p,2}$ and $v_{p,3}$. Hence, there cannot be three vertices in $G'$, which are not adjacent to $w$, which contradicts the fact that $w$ is the isolated vertex of an induced co-claw.

\textbf{Case 2:} $w \in V$. Since $G$ has no induced co-claw, we assume without loss of generality that~$u_1$ lies in $\mathcal{K}$. Then, by construction of $G'$ there is an edge $\lbrace w, u_1 \rbrace \in E'$, which contradicts the fact that $w$ is the isolated vertex of an induced~co-claw.

$(b)$ To show that $G \mapsto (G',k)$ is a correct reduction from \textsc{3-Clique Cover} to \textsc{STC}, we prove that $G$ has a clique cover of size three if and only if there is an STC-labeling $L=(S_L,W_L)$ for $G'$ with $|S_L| \geq 3 \cdot \binom{n^3}{2} + n^4$.

Let $G$ have a clique cover of size three. Then, there are three disjoint cliques $V_1, V_2, V_3$ in $G$ such that $V_1 \cup V_2 \cup V_3 = V$.  We define an STC-labeling $L=(S_L,W_L)$ on $G'$ with at least $(3 \cdot \binom{n^3}{2}+n^4)$ strong edges by setting $S_L := \bigcup_{i=1}^3 E_{G'}(K_i \cup V_i)$. Since all $K_i \cup V_i$ are disjoint cliques, $L$ is an STC-labeling. Moreover, there are at least $(3 \cdot \binom{n^3}{2}+n^4)$ edges in $S_L$ since
\begin{align*}
\sum_{i=1}^3|E_{G'}(V_i \cup K_i)| &= \sum_{i=1}^3 \left( \binom{|V_i|}{2} + \binom{n^3}{2} + |V_i| \cdot n^3 \right)\\
&\geq 3 \cdot \binom{n^3}{2} + n^3 \sum_{i=1}^3 |V_i|\\
&= 3 \cdot \binom{n^3}{2} + n^4.
\end{align*}

Conversely, let $L=(W_L, S_L)$ be an STC-labeling for $G'$ with $|S_L| \geq 3 \cdot \binom{n^3}{2} + n^4$. Assume towards a contradiction that $G$ does not have a clique cover of size three or less. We start with the following observation.

\begin{Observation} \label{Observation about Union of Ki}
There are at most $3 \cdot \binom{n^3}{2}$ strong edges in $E_{G'}(\mathcal{K})$.
\end{Observation}

\begin{proof}[of Observation \ref{Observation about Union of Ki}]
We first show that each vertex $v \in \mathcal{K}$ has at most $n^3-1$ strong neighbors in~$\mathcal{K}$. Without loss of generality assume that $v=v_{1,1} \in K_1$. By the construction of $G'$, the vertex $v_{1,1}$ has exactly $3 \cdot (n^3-1)$ neighbors in $\mathcal{K}$ since it is adjacent to all of $\mathcal{K}$ except $v_{1,2}$, $v_{1,3}$ and itself. Assuming $v_{1,1}$ has more than $n^3-1$ strong neighbors, it follows by the pigeonhole principle that there is a number $d = 2, \dots ,n^3$ such that two vertices $v_{d,i}$ and $v_{d,j}$ with $i \neq j$ are strong neighbors of $v_{1,1}$. Since $\lbrace v_{d,i}, v_{d,j} \rbrace \not \in E'$, the edges $\lbrace v_{d,i}, v_{1,1} \rbrace$ and $\lbrace v_{1,1}, v_{d,j} \rbrace$ form a strong $P_3$ under $L$, which contradicts the fact that $L$ satisfies~STC.

Since $|\mathcal{K}|= 3 \cdot n^3$ and each vertex of $\mathcal{K}$ has at most $n^3-1$ strong neighbors in~$\mathcal{K}$, there are at most $\frac{3 n^3 \cdot (n^3-1)}{2} = 3 \cdot \binom{n^3}{2}$ strong edges between vertices of $\mathcal{K}$. $\hfill \Diamond$
\end{proof}


\begin{Observation} \label{Observation about the edges between}
There are at most $n^4 - n^3$ strong edges in $E_{G'}(\mathcal{K},V)$.
\end{Observation}

\begin{proof} [of Observation \ref{Observation about the edges between}]
For $v_{c,i} \in \mathcal{K}$, let $\mathcal{N}(v_{c,i})= \lbrace w \in V \mid \lbrace v_{c,i}, w \rbrace \in S_L \rbrace$ be the set of strong neighbors of $v_{c,i}$ that lie in $V$. Obviously, each $\mathcal{N}(v_{c,i})$ forms a clique, since otherwise there would be a strong $P_3$ under $L$.

Consider a triple of vertices $v_{c,1},v_{c,2}, v_{c,3} \in V$ for some fixed $c = 1, \dots , n^3$. By the construction of $G'$, those three vertices are pairwise non-adjacent. It follows that $\mathcal{N}(v_{c,1})$, $\mathcal{N}(v_{c,2})$ and $\mathcal{N}(v_{c,3})$ are pairwise disjoint. Otherwise, if there is a vertex $w \in \mathcal{N}(v_{c,1}) \cap \mathcal{N}(v_{c,2})$, the edges $\lbrace v_{c,1}, w \rbrace$ and $\lbrace w, v_{c,2} \rbrace$ form a strong $P_3$ under~$L$.

Since $\mathcal{N}(v_{c,1})$, $\mathcal{N}(v_{c,2})$, and $\mathcal{N}(v_{c,3})$ are disjoint cliques in $V$, the assumption that there is no clique cover of size at most three leads to the fact that for each triple $v_{c,1}, v_{c,2}, v_{c,3}$ we can find a vertex $w \in V$ such that $w \not \in \mathcal{N}(v_{c,1}) \cup \mathcal{N}(v_{c,2}) \cup \mathcal{N}(v_{c,3})$. It follows, that there are at most $n-1$ strong edges in $E_{G'}(V, \lbrace v_{c,1}, v_{c,2}, v_{c,3} \rbrace)$. Since $\mathcal{K}$ consists of exactly $n^3$ such triples, there are at most $n^3 \cdot (n-1)= n^4 - n^3$ strong edges in $E_{G'}(\mathcal{K}, V)$. $\hfill \Diamond$
\end{proof}
Observations~\ref{Observation about Union of Ki} and~\ref{Observation about the edges between} together with the fact that $|E_{G'}(V)|=|E|=\binom{n}{2}< n^3$ gives us the following inequality:
\begin{align*}
|S_L| &= |S_L \cap E_{G'}(\mathcal{K})| + |S_L \cap E_{G'}(\mathcal{K}, V)| + |S_L \cap E_{G'}(V)|\\
&\leq \, 3 \cdot \binom{n^3}{2} + (n^4 - n^3) + \binom{n}{2}\\
&< \, 3 \cdot \binom{n^3}{2} + n^4 \text{.}
\end{align*}
This inequality contradicts the fact that $|S_L| \geq 3 \cdot \binom{n^3}{2} + n^4$. Hence, $G$ has a clique cover of size three or less, which proves the correctness of the reduction.

$(c)$ To show that $G \mapsto (G',k)$ is a correct reduction from \textsc{3-Clique Cover} to \textsc{CD} we need to prove that $G$ has a clique cover of size three if and only if there is a cluster-subgraph $\mathcal{G}'=(V',\mathcal{E}')$ of $G'$ such that $|\mathcal{E}'| \geq 3 \cdot \binom{n^3}{2} + n^4$.

Let $G$ have a clique cover of size three. Since the labeling $L$ from the proof of~$(b)$ is a cluster labeling, there must be a solution for \textsc{CD} with at least $3 \cdot \binom{n^3}{2} + n^4$ cluster edges.

Now, let $\mathcal{G}'=(V',\mathcal{E}')$ be a cluster subgraph of $G'$ such that $|\mathcal{E}'| \geq 3 \cdot \binom{n^3}{2} + n^4$. We define an STC-labeling $L=(W_L,S_L)$ for $G'$ by $S_L:=\mathcal{E}'$. Then, there is an STC-labeling for $G'$ with at least $3 \cdot \binom{n^3}{2} + n^4$ strong edges. By $(b)$, $G$ has a clique cover of size three or less, which proves the correctness of the reduction.

From $(a)$, $(b)$, and $(c)$ we conclude that \textsc{STC} and \textsc{CD} remain NP-hard on co-claw free graphs. \qed
\end{proof}

From Lemmas \ref{poly-time solvable}, \ref{NP-hardness}, and \ref{co-claw-free hardness} we obtain the following theorem.

\begin{theorem} 
The problems \textsc{CD} and \textsc{STC} are
\begin{enumerate}
\item[•] solvable in polynomial time on $H$-free graphs, if $H \in \lbrace K_3, P_3, K_2 + K_1, P_4, \text{paw} \rbrace$, and
\item[•]  NP-hard on $H$-free graphs, if $H \in \lbrace 3K_1, K_4, 4K_1, C_4, 2K_2, \text{diamond}, \text{co-diamond},$\\
$\text{claw},\text{co-claw}, \text{co-paw} \rbrace$.
\end{enumerate}
\end{theorem}

\section{Outlook}
\label{sec:outlook}
Many open questions remain. For example, it is open whether \textsc{Strong Triadic Closure} can
be solved in~$2^{\Oh(\ell)}\cdot \poly(n)$ time. Even an algorithm with running
time~$2^{\Oh(n)}$ is not known at the moment.  For a generalization of~\textsc{STC}, where each vertex has a list of possible incident strong colors, there is no algorithm that solves the problem in~$2^{o(n^2)}$ time if the exponential time hypothesis (ETH) is true~\cite{BGKS19}. It is open if this lower bound can be transferred to~\textsc{STC}.

Furthermore, it is open whether we can
solve~\textsc{Strong Triadic Closure} faster than in~$\mathcal{O}(1.28^k+nm)$, the running time
that is implied by the parameter-preserving reduction to \textsc{Vertex Cover}. It seems that
any faster algorithm would need to use new insights into \textsc{Strong Triadic Closure}.
Moreover, a complete characterization of the graphs in which no optimal \textsc{Cluster
  Deletion} solution is an optimal solution for \textsc{Strong Triadic Closure} is open. Such a
characterization would be also interesting from the application point of view as it would
describe when triadic closure gives a different model than clustering. Concerning
approximability, a factor-2 approximation for minimizing the number of weak edges in \textsc{Strong Triadic Closure} is implied by
the reduction to \textsc{Vertex Cover}. It is open whether there is a polynomial-time
approximation algorithm with a factor smaller than~$2$. It is also open whether optimal solutions for
\textsc{Cluster Deletion} give a constant-factor approximation for the minimization variant of
\textsc{Strong Triadic Closure}. Finally, we restate the following open question of Golovach
et al.~\cite{Pinar}: is \textsc{Strong Triadic Closure} fixed-parameter tractable when
parameterized by~$\ell-|M|$, where~$\ell$ is the number of strong edges and~$M$ is a maximum
matching in the input graph?



\bibliographystyle{plain} 

\begin{thebibliography}{10}

\bibitem{BD11}
Sebastian B{\"{o}}cker and Peter Damaschke.
\newblock Even faster parameterized cluster deletion and cluster editing.
\newblock {\em Inf. Process. Lett.}, 111(14):717--721, 2011.

\bibitem{BodlaenderJK14}
Hans~L. Bodlaender, Bart M.~P. Jansen, and Stefan Kratsch.
\newblock Kernelization lower bounds by cross-composition.
\newblock {\em {SIAM} J. Discrete Math.}, 28(1):277--305, 2014.

\bibitem{BodlaenderTY11}
Hans~L. Bodlaender, St{\'{e}}phan Thomass{\'{e}}, and Anders Yeo.
\newblock Kernel bounds for disjoint cycles and disjoint paths.
\newblock {\em Theor. Comput. Sci.}, 412(35):4570--4578, 2011.

\bibitem{BLS99}
Andreas Brandstädt, Van~Bang Le, and Jeremy~P Spinrad.
\newblock {\em Graph classes: a survey}, volume~3.
\newblock SIAM, 1999.

\bibitem{BGKS19}
Laurent Bulteau, Niels Gr{\"{u}}ttemeier, Christian Komusiewicz, and Manuel
  Sorge.
\newblock Your rugby mates don't need to know your colleagues: Triadic closure
  with edge colors.
\newblock In {\em Proceedings of the 11th International Conference on
  Algorithms and Complexity ({CIAC}~'19)}, volume 11485 of {\em Lecture Notes
  in Computer Science}, pages 99--111. Springer, 2019.

\bibitem{Cyg+15}
Marek Cygan, Fedor~V. Fomin, Lukasz Kowalik, Daniel Lokshtanov, D{\'{a}}niel
  Marx, Marcin Pilipczuk, Michal Pilipczuk, and Saket Saurabh.
\newblock {\em Parameterized Algorithms}.
\newblock Springer, 2015.

\bibitem{DF13}
Rodney~G. Downey and Michael~R. Fellows.
\newblock {\em Fundamentals of Parameterized Complexity}.
\newblock Texts in Computer Science. Springer, 2013.

\bibitem{GHN13}
Yong Gao, Donovan~R. Hare, and James Nastos.
\newblock The cluster deletion problem for cographs.
\newblock {\em Discrete Math.}, 313(23):2763--2771, 2013.

\bibitem{Pinar}
Petr~A. Golovach, Pinar Heggernes, Athanasios~L. Konstantinidis, Paloma~T.
  Lima, and Charis Papadopoulos.
\newblock Parameterized aspects of strong subgraph closure.
\newblock In {\em Proceedings of the 16th Scandinavian Symposium and Workshops
  on Algorithm Theory, (SWAT '18)}, volume 101 of {\em Leibniz International
  Proceedings in Informatics (LIPIcs)}, pages 23:1--23:13. Schloss Dagstuhl -
  Leibniz-Zentrum fuer Informatik, 2018.

\bibitem{Granovetter73}
Mark~S Granovetter.
\newblock The strength of weak ties.
\newblock {\em Am. J. Sociol.}, 78:1360--1380, 1973.

\bibitem{GK18}
Niels Gr{\"{u}}ttemeier and Christian Komusiewicz.
\newblock On the relation of strong triadic closure and cluster deletion.
\newblock In {\em Proceedings of the 44th International Workshop
  Graph-Theoretic Concepts in Computer Science ({WG}~'18)}, volume 11159 of
  {\em Lecture Notes in Computer Science}, pages 239--251. Springer, 2018.

\bibitem{Guo09}
Jiong Guo.
\newblock A more effective linear kernelization for cluster editing.
\newblock {\em Theor. Comput. Sci.}, 410(8-10):718--726, 2009.

\bibitem{Holyer}
Ian Holyer.
\newblock The {NP}-completeness of edge-coloring.
\newblock {\em {SIAM} J. Comput.}, 10(4):718--720, 1981.

\bibitem{HsuL}
Wen{-}Lian Hsu and Tze{-}Heng Ma.
\newblock Substitution decomposition on chordal graphs and applications.
\newblock In {\em Proceedings of the 2nd International Symposium on Algorithms
  (ISA~'91)}, volume 557 of {\em Lecture Notes in Computer Science}, pages
  52--60. Springer, 1991.

\bibitem{KU12}
Christian Komusiewicz and Johannes Uhlmann.
\newblock Cluster editing with locally bounded modifications.
\newblock {\em Discrete Appl. Math.}, 160(15):2259--2270, 2012.

\bibitem{KNP17L}
Athanasios~L. Konstantinidis, Stavros~D. Nikolopoulos, and Charis Papadopoulos.
\newblock Strong triadic closure in cographs and graphs of low maximum degree.
\newblock {\em Theor. Comput. Sci.}, 740:76--84, 2018.

\bibitem{konstantinidis_et_alL}
Athanasios~L. Konstantinidis and Charis Papadopoulos.
\newblock {Maximizing the Strong Triadic Closure in Split Graphs and Proper
  Interval Graphs}.
\newblock In {\em Proceedings of the 28th International Symposium on Algorithms
  and Computation (ISAAC~'17)}, volume~92 of {\em Leibniz International
  Proceedings in Informatics (LIPIcs)}, pages 53:1--53:12, Dagstuhl, Germany,
  2017. Schloss Dagstuhl--Leibniz-Zentrum fuer Informatik.

\bibitem{Bang96}
Van~Bang Le.
\newblock Gallai graphs and anti-gallai graphs.
\newblock {\em Discrete Math.}, 159(1-3):179--189, 1996.

\bibitem{MV80L}
Silvio Micali and Vijay~V. Vazirani.
\newblock An ${O(\sqrt{|V|}|E|)}$ algorithm for finding maximum matching in
  general graphs.
\newblock In {\em Proceedings of the 21st Annual Symposium on Foundations of
  Computer Science (FOCS~'80)}, pages 17--27. {IEEE} Computer Society, 1980.

\bibitem{N99}
Assaf Natanzon.
\newblock Complexity and approximation of some graph modification problems.
\newblock Master thesis, University of Tel-Aviv, 1999.

\bibitem{olariu1988paw}
Stephan Olariu.
\newblock Paw-free graphs.
\newblock {\em Inf. Process. Lett.}, 28(1):53--54, 1988.

\bibitem{poljak1974}
Svatopluk Poljak.
\newblock A note on stable sets and colorings of graphs.
\newblock {\em Commentationes Mathematicae Universitatis Carolinae},
  15(2):307--309, 1974.

\bibitem{prisner1993}
Erich Prisner.
\newblock Hereditary clique-helly graphs.
\newblock {\em Journal of Combinatorial Mathematics and Combinatorial
  Computing}, 14:216--220, 1993.

\bibitem{Protti2009}
F{\'a}bio Protti, Maise Dantas~da Silva, and Jayme~Luiz Szwarcfiter.
\newblock Applying modular decomposition to parameterized cluster editing
  problems.
\newblock {\em Theory Comput. Syst.}, 44(1):91--104, Jan 2009.

\bibitem{RTG17}
Polina Rozenshtein, Nikolaj Tatti, and Aristides Gionis.
\newblock Inferring the strength of social ties: {A} community-driven approach.
\newblock In {\em Proceedings of the 23rd {ACM} {SIGKDD} International
  Conference on Knowledge Discovery and Data Mining (KDD~'17)}, pages
  1017--1025. {ACM}, 2017.

\bibitem{bihi80}
Najiba Sbihi.
\newblock Algorithme de recherche d'un stable de cardinalite maximum dans un
  graphe sans etoile.
\newblock {\em Discrete Math.}, 29(1):53--76, 1980.

\bibitem{SST04}
Ron Shamir, Roded Sharan, and Dekel Tsur.
\newblock Cluster graph modification problems.
\newblock {\em Discrete Appl. Math.}, 144(1-2):173--182, 2004.

\bibitem{sintosL}
Stavros Sintos and Panayiotis Tsaparas.
\newblock Using strong triadic closure to characterize ties in social networks.
\newblock In {\em Proceedings of the 20th ACM SIGKDD International Conference
  on Knowledge Discovery and Data Mining (KDD~'14)}, pages 1466--1475, New
  York, NY, USA, 2014. ACM.

\bibitem{Sun91}
Liping Sun.
\newblock Two classes of perfect graphs.
\newblock {\em J. Comb. Theory, Ser. {B}}, 53(2):273--292, 1991.

\bibitem{BTZ19}
Ren{\'{e}} van Bevern, Oxana~Yu. Tsidulko, and Philipp Zschoche.
\newblock Fixed-parameter algorithms for maximum-profit facility location under
  matroid constraints.
\newblock In {\em Proceedings of the 11th International Conference on
  Algorithms and Complexity ({CIAC}~'19)}, volume 11485 of {\em Lecture Notes
  in Computer Science}, pages 62--74. Springer, 2019.

\end{thebibliography}

\end{document}